\documentclass[preprint]{imsart}

\pdfoutput=1

\usepackage[font=small]{caption}
\usepackage{graphicx}
\usepackage{caption}
\usepackage{subcaption}
\usepackage{multirow}

\usepackage{latexsym}
\usepackage{amsmath}
\usepackage{amsfonts}
\usepackage{amssymb}
\usepackage{psfrag}

\RequirePackage{natbib}
\RequirePackage[colorlinks,citecolor=blue,urlcolor=blue]{hyperref}

\usepackage{booktabs}

\startlocaldefs

\def\negr#1{\mbox{\boldmath$#1$}}
\newcommand{\matr}[1]{\mathbf{#1}}
\newcommand{\vect}[1]{\mathbf{#1}}
\newcommand{\el}{ {\cal{L}}}
\newcommand{\E}{ {\rm{E}}}
\newcommand{\I}{ {\rm{I}}}

\newcommand{\pzk} {p_{k}(\mathbf{s})}

\newcommand{\allf} {{\vect{f}}}

\newcommand{\yrep}{ \mathbf{y}^{R}  }
\newcommand{\zrep}{ \mathbf{z}^{R}  }
\newcommand{\nbasis}{ K  }
\newcommand{\lnbasis}{ \nu  }
\newcommand{\minitab}[2][l]{\begin{tabular}{#1}#2\end{tabular}}

\newcommand{\thickhline}{%
    \noalign {\ifnum 0=`}\fi \hrule height 1pt
    \futurelet \reserved@a \@xhline
}

\endlocaldefs

\begin{document}

%\SetWatermarkLightness{0.9}
%\SetWatermarkScale{4}

\begin{frontmatter}

% "Title of the paper"
\title{Switching nonparametric regression models for multi-curve data}
\runtitle{Switching nonparametric regression models for multi-curve data }

% indicate corresponding author with \corref{}
% \author{\fnms{John} \snm{Smith}\corref{}\ead[label=e1]{smith@foo.com}\thanksref{t1}}
% \thankstext{t1}{Thanks to somebody}
% \address{line 1\\ line 2\\ printead{e1}}
% \affiliation{Some University}

%\author{Camila P. E. de Souza\authorref{1,}\authorref{2}\thanksref{t1}}
%\author{Nancy E. Heckman\authorref{3}}
%\author{Fan Xu\authorref{4}}
%\affiliation[1]{Department of Pathology and Laboratory Medicine, University of British Columbia, Vancouver, Canada}
%\affiliation[2]{Department of Molecular Oncology, BC Cancer Agency, Vancouver, Canada}
%\affiliation[3]{Department of Statistics, Univerisity of British Columbia, Vancouver, Canada}
%\affiliation[4]{Department of Industrial Engineering and Operations Research, Columbia University, New York, United States}

\author{\fnms{Camila} \snm{P. E. de Souza,$^{1,2}$}\thanksref{t1}\ead[label=e1]{desouzacpe@gmail.com}}
\author{\fnms{Nancy} \snm{E. Heckman$^3$}\thanksref{t1}\ead[label=e2]{nancy@stat.ubc.ca}}

\and
\author{\fnms{Fan} \snm{Xu$^4$}\thanksref{t1}\ead[label=e3]{xufan9999@gmail.com}}
\address{$^1$Department of Pathology and Laboratory Medicine, University of British Columbia, Vancouver, Canada \\
$^2$Department of Molecular Oncology, BC Cancer Agency, Vancouver, Canada\\
$^3$Department of Statistics, Univerisity of British Columbia, Vancouver, Canada \\
$^4$Department of Industrial Engineering and Operations Research, Columbia University, New York, United States \\
\printead{e1}\\
%\printead*{e2}\\
%\printead*{e3}
%\phantom{E-mail:\ }
}
%\affiliation{University of British Columbia}

\thankstext{t1}{Research supported by the National Science and Engineering Research Council of Canada.}

\runauthor{De Souza et al.}

\begin{abstract}
We develop and apply an approach for analyzing multi-curve data where each curve is driven by a latent state process. The state at any particular point determines a smooth function, forcing the individual curve to ``switch" from one function to another.  Thus each curve follows what we call a switching nonparametric regression model.  We develop an EM algorithm to estimate the model parameters. We also obtain standard errors for the parameter estimates of the state process. We consider three types of hidden states, those that are independent and identically distributed, those that follow a Markov structure and those that are independent but with distribution depending on some covariate(s).  A simulation study shows the frequentist properties of our estimates. We apply our methods to a building's power usage data.
\end{abstract}

%\begin{keyword}[class=AMS]
%\kwd[Primary ]{}
%\kwd{}
%\kwd[; secondary ]{}
%\end{keyword}

%\begin{keyword}
%\kwd{}
%\kwd{}
%\end{keyword}


\end{frontmatter}

% AOS,AOAS: If there are supplements please fill:
%\begin{supplement}[id=suppA]
%  \sname{Supplement A}
%  \stitle{Title}
%  \slink[doi]{10.1214/00-AOASXXXXSUPP}
%  \sdatatype{.pdf}"
%  \sdescription{Some text}
%\end{supplement}

\section{Introduction} \label{sec:motivation}

We develop and apply a method for analyzing multi-curve data where each curve follows a {\em{switching nonparametric regression model }}  \citep{deSouza2014}. That is, each curve,  over its domain, switches among $J$ unobserved states with each state determining a function.  The main goal is to estimate the function corresponding to each state and the parameters of the latent process, along with some measure of accuracy.

We are motivated by the problem of calculating a building's ``typical curve" of energy consumption, that is, its expected energy consumption as a function of time and other variables (e.g., weather conditions).   Such knowledge  allows building managers to compare the building's real-time performance to its ``typical" performance which is useful, for instance,   for assessing  the impact of improvements on a building's energy efficiency.   The data set we analyze was provided by PulseEnergy, now part of EnerNOC (www.enernoc.com).

To understand our methodological approach, compare the plots in Figure \ref{building_figures}.
Figure \ref{building_one} shows hourly power usage during the months of June and July 2009 in an office building. On some days (holidays and weekends) energy usage is close to zero. We observe that on some business days the energy usage is very high, approximately twice as much as on the other days.  This high power consumption occurs on warm days, when the cooling system (also called the chiller) of the building was probably on. Figure \ref{building_replicates} presents the building daytime power usage from 9am to 4pm for 44 business days in June and July 2009. Several types of curves can be observed: one type corresponds to days when the cooling system was probably on and another type when the cooling system was off. We also observe that on some days the chiller turned on in the middle of the day. On one day the chiller went on, off and then on again.

\citet{brown2012kernel} consider the data in Figure \ref{building_one} using a very computer intensive method.  They find the ``typical curve" by applying a local constant kernel smoother over an extremely  large number of data points, and thus, their contribution to the analysis is mainly on improving computational efficiency.  They do not consider the special structure we see in Figure \ref{building_replicates}.  One shortfall of their smoothing method is that  they do not model 
the abrupt changes in level of energy consumption, and thus their approach may oversmooth these changes.  Since these changes are real features of the data, they should be modelled explicitly to better understand power usage. Our method exploits the  structure of Figure \ref{building_replicates} and differs from the approach proposed by \citet{brown2012kernel}  in two important ways: by treating each business day as a replicate; and by modelling abrupt changes in the building's power usage as arising from two functions, one function giving power usage when the chiller is off, the other function giving power usage when the chiller is on. The condition ``chiller on''/``off'' at any particular time is not recorded by the automatic monitoring system.  Thus, it can only be inferred from the data, and so the state of the chiller forms a latent process.

\citet{deSouza2014} present the case where there is a single realization, a single curve switching among $J$ functions. In that paper, we consider two models for the latent process: one where the states are independent and identically distributed, the other where the sequence of states forms a Markov chain. In addition to estimating all parameters and functions, we derive standard errors for the parameters of the latent process. In the present paper, we extend our 2014 approach into the realm of functional data analysis \citep{ramsay2005}:  we consider the case when there are $N$  curves, called replicates, with each replicate switching among $J$ functions. This is the first work to consider the mixture of multiple functions in functional data analysis. 
We also consider a third type of latent state process, where the state depends on a time-varying covariate. In our application, the covariate is temperature recorded at a weather station several kilometres from the building. Preliminary data analysis indicates this dependence can be modelled via logistic regression.

Several authors have considered the single realization case from a Bayesian perspective with the smooth functions modeled as realizations of Gaussian processes.   See, for instance, 
\citet{tresp}, \citet{rasmussen2002} and \citet{ou2008}. The paper of \citet{ou2008} also contains a Bayesian analysis of the replicate case. These papers are discussed in more detail in \citet{deSouza2014} and contain methodology that can, in principle, lead to estimation of all $J$ functions and the latent variable process parameters.  However, unlike our work, the focus is on the estimation of just one function - the mixture, that is, a weighted average of the $J$ functions.  

In a more recent related work, \citet{langrock2014markov} consider generalized additive models with a time component, where the predictor is subject to regime changes controlled by an underlying Markov process. The parameter estimates are obtained by a numerical maximum penalized likelihood approach.  The authors focus on a single realization case and do not consider the replicate case.

This paper is organized as follows. In Section \ref{overview} we provide an overview of the proposed methodology. The solution to the estimation problem is described in Section \ref{sec:parameter_estimation}. Some of the calculations are similar to those that appear in \citet{deSouza2014}; these calculations are given in the Supplementary Material. In Section \ref{sec:sim_studies} we present the results of a simulation study. An application of the proposed methodology to a building's power usage data is presented in Section \ref{data_analysis}. Some discussion is provided in Section \ref{sec:discussion}. 

The computing code and the data are available as supplementary material for possible use by interested readers.

\section{Overview of the proposed methodology}\label{overview}

We consider a data set with $N$ replicates where replicate $k$ contains $n$ observations $y_{1k},\ldots,y_{nk}$ and evaluation points $x_{1},\ldots,x_{n}$, which for simplicity are the same across replicates. Observation $y_{ik}$ depends on $x_i$ according to a hidden (unobserved) state $z_{ik}$ with   possible state values  in $\{1,\ldots,J\}$. If $z_{ik}=j$ the expected response of $y_{ik}$ is $f_j(x_i)$. In this work, we assume the replicates are all generated from just one set of functions $f_1,\ldots,f_J$, a reasonable assumption for the power usage data presented in Figure \ref{building_replicates} and described in Section \ref{sec:motivation}. We consider three types of hidden states, those that are independent and identically distributed, those that follow a Markov structure and those that are independent but with distribution depending on some covariate(s). 

In principal, the $x$s can differ in value and number across replicates.  To proceed, we need only to modify our notation and calculations, since we will model each $f_j$ as a linear combination of B-spline basis functions.  However, in our Markov state process model, a conceptual challenge arises in interpreting transition probabilities when the $x$s vary from replicate to replicate. 

Our notation is as follows.
\begin{itemize}
\item Observed data: $\mathbf{x}=(x_1,\ldots,x_n)^T$, fixed across replicates; 
covariate vectors $\vect{v}_{1k}, \ldots, \vect{v}_{nk}$;  responses  $\yrep=(\vect{y}_1^T,\ldots,$ $\vect{y}_N^T)^T$, where $\vect{y}_k = (y_{1k},\ldots,y_{nk})^T$.
\item Hidden states: $\zrep=(\vect{z}_1^T,\ldots,\vect{z}_N^T)^T$, where  $\vect{z}_k = (z_{1k},\ldots,z_{nk})^T$.
\item $f_j(\mathbf{x})=(f_j(x_1),\ldots, f_j(x_n))^T$ for $j=1,\ldots,J$, and $f_{\vect{z}_k}(\vect{x}) = \big(f_{z_{1k}}(x_1),$ $\ldots,f_{z_{nk}}(x_n)\big)^T$.
\end{itemize}

We assume that $\vect{z}_1,\ldots,\vect{z}_N$, are independent.   Given the hidden states $\vect{z}_k$, $\vect{y}_{k} = f_{\vect{z}_k}
(\vect{x}) + \negr{\epsilon}_{k}$, where $\negr{\epsilon}_1,\ldots,\negr{\epsilon}_N$, are independent and $\negr{\epsilon}_k $ 
has a multivariate normal distribution with mean equal to the 0-vector and covariance matrix $\matr{V}$, possibly depending on  $\vect{z}_k$.  That is,  $\negr{\epsilon}_k \sim MVN(\vect{0},\matr{V})$. 
Therefore, $\vect{y}_1,\ldots, \vect{y}_N$ are independent and, given the hidden states $\vect{z}_k$, $\vect{y}_k \; \sim \; MVN(f_{\vect{z}_{k}}(\vect{x}), \matr{V})$. 
Our model can be considered a functional data model.  In usual functional data modeling, when there is no switching regression, the observations from the $k$th replicate, $y_{1k},\ldots, y_{nk}$, are generated from a single realization of a stochastic process (see, for instance, \citealp{james2000principal}, and \citealp{yao2005functional}).   In our case, for the $k$th replicate, the observations arise from $J$ stochastic process realizations, $f_{1k}, \ldots, f_{Jk}$, one for each possible state. The distribution of the $k$th replicate of the $j$th stochastic process satisfies $E(f_{jk}(x))= f_j(x)$ with the covariance between $f_{jk}(x)$  and $f_{jk}(x^*)$ generating the covariance matrix $\matr{V}$. Thus, $\matr{V}$ induces a dependence among the observations of the $k$th realization. 

We let $\gamma$ be the set containing $f_1(\vect{x}), \ldots, f_J(\vect{x})$ and the parameters in $\matr{V}$.  We assume that the distribution of each $\vect{z}_k$ is governed by a parameter vector $\alpha$.  Section \ref{Sec:choices_alpha_V} presents our different choices of $\matr{V}$ and $\alpha$. 

Our goal is to estimate $\theta \equiv \{\alpha,\gamma\}$, along with standard errors or some measure of accuracy for the parameters in $\alpha$. 
Similar to \citet{deSouza2014} we obtain the parameter estimates by maximizing
\begin{equation}
l(\theta) \equiv \sum_{k=1}^N \log p(\mathbf{y}_k|\theta)  + P(f_1,\ldots,f_J,\lambda_1,\ldots,\lambda_J),
\label{criterion1}
\end{equation}

\noindent where 
$p(\vect{y}_k|\theta)$ is the likelihood function based on the observed data from the $k$th replicate and $P$ is a roughness penalty on the $f_j$s. The exact form of $P(f_1,\ldots,f_J,\lambda_1,\ldots,\lambda_J)$ is chosen by the user. For our work, we set 
\begin{equation}
P(f_1,\ldots,f_J,\lambda_1,\ldots,\lambda_J) = - \sum_{j=1}^J \lambda_j \int [f_j''(x)]^2dx, 
\nonumber
\label{P:PL}
\end{equation}

\noindent since the integrated squared second derivative of a function is a common form of roughness penalty \citep{wahba1990spline}. The $\lambda_j$s are the smoothing parameters, governing the weight of the penalty term.  As in  \citet{deSouza2014} one could also take a Bayesian approach by maximizing (\ref{criterion1}) with $P$ arising from placing a Gaussian process prior on the $f_j$s. 

The form of $\log p(\mathbf{y}_k|\theta)$ is very complicated, since it involves the distribution of the latent states $\vect{z}_{k}$. Therefore, we apply an Expectation-Maximization (EM) algorithm \citep{dempster1977em} to maximize (\ref{criterion1}). We can show (see, for instance, \citealp{cappe2005} and \citealp{mclachlan}) that our EM algorithm generates a sequence of estimates, $\theta^{(c)}$, $c \geq 1$, satisfying  $l(\theta^{(c+1)}) \geq l(\theta^{(c)})$. One could also perform a numerical likelihood maximization as described in \citet{macdonald2014numerical} and \citet{zucchini2016hidden}.

As in the single realization case presented in \citet{deSouza2014} we use again the results of \citet{louis1982} to obtain standard errors for the estimates of the parameters of the latent state process. When the hidden states, $z_{1k},\ldots,z_{nk}$, are independent and identically distributed (\textit{iid}) we consider $J\geq2$ possible state values. For $z_{1k},\ldots,z_{nk}$ following a Markov structure we restrict the possible number of states to $J=2$. We also obtain standard errors for the intercept and slope parameters for the case where $J=2$ and $z_{1k},\ldots,z_{nk}$ are independent with the distribution of $z_{ik}$ depending on only one covariate. See Section 2 of the Supplementary Material for more details. 

\subsection{Choices of $\matr{V}$ and $\alpha$}  \label{Sec:choices_alpha_V}
We consider five models for the covariance of the residual error, $\matr{V}$: unrestricted, diagonal with either $\matr{V}= \sigma^2 \matr{I}$
or with entry $(i,i)$ depending on the latent state, and two generated from a ``random intercept" covariance structure: a {\em{homogeneous random intercept model}} and a {\em{non-homogeneous random intercept model}} with variability of the intercept depending on the value of the latent state. We usually use $\matr{V}_{\vect{z}}$ to denote models where the variability depends on the latent state. However, sometimes we omit the subscript $\vect{z}$ when referring to a general $\matr{V}$. 
The unknown parameters in $\matr{V}$ are clear for our first two models.  For the third model, the parameters in $\matr{V}_{\vect{z}}$ are $\sigma_1^2, \ldots, \sigma_J^2$. 

For $\matr{V}$ to follow a homogeneous random intercept model, let  $y_{ik} = f_{z_{ik}}(x_i) + \epsilon_{ik}$.  Then suppose that $\epsilon_{ik}=\delta_k + e_{ik}$, where $\delta_k$
and $e_{ik}$ are independent for all $i=1,\ldots,n$ and $k=1,\ldots,N$, and the $\delta_k $s are \textit{iid} $N(0,\tau^2)$ and the $e_{ik}$s are \textit{iid} $N(0,\sigma^2)$. Then $\matr{V}$ will depend on only two parameters and can be written as 
\begin{equation}
\matr{V} = \sigma^2 \big(\matr{I} + d \vect{1}\vect{1}^T \big),
\label{restricted_V}
\end{equation}
\noindent where $\matr{I}$ is an $n\times n$ identity matrix, $\vect{1}$ is an $n$-vector of ones and $d = \tau^2/\sigma^2$. 

Our data analysis (Section \ref{data_analysis}) requires the more complex covariance structure of a non-homogenous random intercept model, where the variance of the random intercept depends on the state.  We define this model for the simple case,  where there are $J=2$ states.  We assume that $y_{ik} = f_{z_{ik}}(x_i) + \epsilon_{z_{ik},\,ik}$, where $\epsilon_{1,ik}=\delta_k + e_{ik}$ when $z_{ik} =1 $ and $\epsilon_{2,ik}=\delta_k + \vartheta_k + e_{ik}$ when $z_{ik} =2$. In addition, $\delta_k$, $\vartheta_k$, and $e_{ik}$ are independent for $i=1,\ldots,n$, and $k=1,\ldots,N$, with $\delta_k $s  \textit{iid} $N(0,\tau_1^2)$,  $\vartheta_k $s \textit{iid} $N(0,\tau_2^2)$ and $e_{ik}$s \textit{iid} $N(0,\sigma^2)$. Therefore, the covariance matrix for the non-homogeneous random intercept model is given by
\begin{equation}
\matr{V}_{\vect{z}_k} = \sigma^2 ( \matr{I} + d_1\matr{1}\matr{1}^T + d_2\matr{1}_{\vect{z}_k}\matr{1}_{\vect{z}_k}^T),
\label{covEq:RI:diff}
\end{equation}

\noindent where  $d_j= \tau_j^2/\sigma^2$ and $\matr{1}_{\vect{z}_k}$ is an $n$-vector with $i$th entry $\I(z_{ik}=2)$. 

In our model $\alpha$ is the vector containing the parameters governing the distribution of the hidden states. If $z_{1k},\ldots,z_{nk}$ are \textit{iid}, then  $\alpha$ is of length $J$ with $j$th component equal to $p(z_{ik}=j | \alpha) \equiv p_j$. If $z_{1k},\ldots,z_{nk}$ follow a Markov structure, that is, if
$ p(z_{ik} | z_{(i-1)k},\ldots,z_{1k},\alpha)=p(z_{ik} | z_{(i-1)k},\alpha)$,  $i=2,\ldots n$, then the parameter vector $\alpha$ consists of the initial probabilities, $\pi_j=p(z_{ik}=1|\alpha)$,  and the transition probabilities, 
$  a_{lj} =  
p(z_{ik}=j|z_{(i-1)k}=l,\alpha) $, $j, l=1,\ldots,J$.
Note that 
 the transition probabilities do not depend on $i$ or $k$.

In the case where $z_{1k},\ldots,z_{nk}$ are independent, with the distribution of $z_{ik}$ depending on a vector of covariates $\vect{v}_{ik} = (1,v_{1,ik},v_{2,ik}, \ldots,v_{M,ik})^T$, we assume that $p(z_{ik} = j | \vect{v}_{ik}, \alpha) \equiv p_j(\vect{v}_{ik}, \alpha)$ follows a multinomial logistic regression model with 
$$\log \frac{ p_j(\vect{v}_{ik}, \alpha)} { p_1(\vect{v}_{ik}, \alpha)} = \beta_{j0} + \beta_{j1}v_{1,ik} + \cdots +  \beta_{jM}v_{M,ik} = \negr{\beta}_j^T \vect{v}_{ik} \; \mbox{for} \; j=2,\ldots,J $$

\noindent so that 
$$  p_1(\vect{v}_{ik},\alpha) = \frac{1}{1+ \sum_{j=2}^{J}e^{ \negr{\beta}_j ^T\vect{v}_{ik} }} $$
\noindent and
$$ \; p_j(\vect{v}_{ik},\alpha) = \frac{e^{ \negr{\beta}_j^T \vect{v}_{ik} }}{1+ \sum_{j=2}^{J}e^{ \negr{\beta}_j^T \vect{v}_{ik} }} \; \mbox{for} \; j=2,\ldots,J.$$ 
In this case $\alpha$ contains all the regression coefficient vectors  $\negr{\beta}_2, \ldots, \negr{\beta}_{J}$.

\section{Parameter estimation}{\label{sec:parameter_estimation}}
Here we present the proposed EM algorithm to obtain the estimates of the parameters in $\theta$. In the M-step, we take the same approach as \citet{deSouza2014}  and model each $f_j$ as a linear combination of $K$ known cubic B-spline basis functions, so that $f_j(\vect{x})=\matr{B}\phi_j$, where $\phi_j$ is the $\nbasis$-vector of coefficients corresponding to $f_j$ and $\matr{B}$ is the  $n\times \nbasis$ matrix with entries $B_{i\lnbasis}=b_\lnbasis(x_i)$. 

The smoothing parameters, $\lambda_1,\ldots,\lambda_J$, can be chosen by a data driven method or subjectively by visual inspection. In Section \ref{choice:lambda:rep}, we propose and justify a leave-one-curve-out cross-validation criterion to find the optimal $\lambda_j$s for the case  when $\matr{V}$ is diagonal and use this method in our application. In our application, when $\matr{V}$ is based on the nonhomogeneous random intercept model, we choose the smoothing parameters via a ``brute force" leave-one-curve-out method, assuming that $\lambda_1=\lambda_2=\lambda$. We use a weighted cross-validation criterion where the weights reflect the uncertainty of the hidden states (see Section 3 of the Supplementary Material).  In all of our simulation studies, to reduce computation time, we pre-choose the $\lambda_j$s by examining a few data sets and visually ensuring that the estimated functions have the same smoothness and shape as the true curves.

Let $p(\yrep,\zrep | \theta) $ be the joint distribution of the observed and latent data given $\theta$, also called the complete data distribution. The application of the EM algorithm to the replicate case is  similar to that of the one realization case considered in \citet{deSouza2014}, which is based on writing 
\begin{equation}
\log p(\yrep,\zrep | \theta) =
\log p(\yrep|\zrep,\theta)  +  \log p(\zrep|\theta)
\equiv
\el_{1}(\gamma)   + 
\el_{2}(\alpha) .
% \label{eq:L1L2}
\nonumber
\end{equation}
In what follows we present a summary of the E and M steps. See Section 1 of the Supplementary Material for details.

In the E-step we calculate 
\begin{equation}
Q(\theta,\theta^{(c)}) 
\equiv  \E_{\theta^{(c)}} \big( \log p(\yrep,\zrep |\theta)|\yrep \big)
 =  \E_{\theta^{(c)}}(\el_1(\gamma)|\yrep) + \E_{\theta^{(c)}}(\el_2(\alpha)|\yrep). \nonumber
\label{eq:Q}
\end{equation}

In the M-step, we want to find $\theta^{(c+1)}$ that maximizes $S(\theta,\theta^{(c)}) \equiv Q(\theta,\theta^{(c)}) + P(f_1,\ldots,f_J,\lambda_1,\ldots,\lambda_J)$ with respect to $\theta$, or at least satisfies  $S(\theta^{(c+1)},\theta^{(c)})$ $ \ge S(\theta^{(c)},\theta^{(c)})$.
Let $\vect{s}$ be an $n$-vector of possible hidden states, i.e., each entry of $\vect{s}$ is in $\{1, 2, \ldots,J\}$, and let
\[
\pzk^{(c)}\equiv p(\vect{z}_k=\vect{s}|\yrep,\theta^{(c)}) 
=  p(\vect{z}_k=\vect{s}|\vect{y}_k,\theta^{(c)}),\]

\noindent whose value depends on the model assumed for the hidden states. Therefore, disregarding the constant terms, we maximize
\begin{eqnarray}
 &&S^* (\theta,\theta^{(c)})  \equiv
\nonumber \\
&  - &  
 \displaystyle \frac{1}{2} \sum_{k=1}^{N}  \sum_{\text{all} \, \vect{s}}  \pzk^{(c)}
\left[
(\vect{y}_k - f_{\vect{s}}(\vect{x}) )^T\matr{V}_{\vect{s}}^{-1} (\vect{y}_k - f_{\vect{s}}(\vect{x})) 
+ \log |\matr{V}_{\vect{s}} | \right ] \label{S:rep:1} 
\\
&  +  &~  P(f_1,\ldots,f_J,\lambda_1,\ldots,\lambda_J) \label{S:rep:2}
\\
& + &~   \sum_{k=1}^N  \sum_{\text{all} \; \vect{s}}  \pzk^{(c)}  \log  p(\vect{z}_k = \vect{s} | \alpha)
 \label{S:rep:3}
\end{eqnarray}

\noindent with respect to $\theta$ $=\{\alpha, f_1(\vect{x}),\ldots,f_J(\vect{x})$, and the parameters in $\matr{V} \}$. Note that $\theta^{(c)}$ is fixed and
thus so are the $\pzk^{(c)}$s. We also consider the smoothing parameters, $\lambda_1,\ldots, \lambda_J$, to be fixed. We apply a natural extension of the EM approach, the Expectation-Conditional Maximization (ECM) algorithm \citep{meng}, to obtain the parameter updates $\theta^{(c+1)}$.

\subsection{M-step via an ECM algorithm}\label{Mstep:rep:iid}
The steps of the ECM algorithm are summarized as follows.
\begin{enumerate}
\item Hold $\matr{V}$ and the parameters in $\alpha$ fixed and maximize $S^*$ with respect to $f_1(\vect{x}),\ldots,f_J(\vect{x})$, obtaining $f_1(\vect{x})^{(c+1)},\ldots,f_J(\vect{x})^{(c+1)}$. That is, maximize the sum of (\ref{S:rep:1}) and (\ref{S:rep:2}).
\item Hold $f_1(\vect{x}),\ldots,f_J(\vect{x})$ and the parameters in $\alpha$ fixed and maximize (\ref{S:rep:1}) with respect to the parameters in $\matr{V}$, obtaining $\matr{V}^{(c+1)}$.
\item Hold $f_1(\vect{x}),\ldots,f_J(\vect{x})$ and $\matr{V}$ fixed and maximize (\ref{S:rep:3}) with respect to the parameters in $\alpha$, obtaining $\alpha^{(c+1)}$.
\end{enumerate}

The results for Steps 1, 2 and 3 are given below. Details can be found in Section 1.2 of the Supplementary Material.

\vspace{.3cm}
\noindent Step 1. \textit{Updating $f_1(\vect{x}),\ldots,f_J(\vect{x})$}.

We propose a method to update the $f_j(\vect{x})$s that is straightforward and yields an estimate of $\allf= (f_1(\vect{x})^T,\ldots,f_J(\vect{x})^T)^T$ in closed form.  The trick is to write $f_{\vect{s}}(\vect{x})$  in terms of $f_1(\vect{x}),\ldots,f_J(\vect{x})$. 
To do this, let  $\vect{1}_{j,\vect{s}}$ be the $n$-vector with $i$th element  equal to 1 if $s_i=j$,  0 else.   Let ${\mathcal{I}}_{\vect{s}}$ be the $n$ by $nJ$ matrix, ${\mathcal{I}}_{\vect{s}}= [ $ diag$(\vect{1}_{1,\vect{s}})  ~ |  ~  \cdots   ~ |  ~  $diag$(\vect{1}_{J,\vect{s}})]$.  Then we easily see that  $f_{\vect{s}}(\vect{x}) = {\mathcal{I}}_{\vect{s}} \allf$. Recall that  $f_j(\vect{x}) = \matr{B} \phi_j$. 
 Let $\matr{B}^{*}$ be the $nJ\times \nbasis J$ block diagonal matrix with each block equal to $\matr{B}$ and  let $\phi$ be the $JK$-vector $\phi=(\phi_1^T,\ldots,\phi_J^T)^T$.  Therefore $\allf = \matr{B}^{*}\phi$.
Let $\matr{R}$ be the  $\nbasis \times \nbasis$ matrix with entries 
$
\matr{R}_{\lnbasis \lnbasis'}= \int b_\lnbasis''(x)b''_{\lnbasis'}(x)\,dx.
$
Combining these calculations we see that, to find the $f_j$s that maximize the sum of (\ref{S:rep:1}) and (\ref{S:rep:2}), we must maximize, as a function of $\phi$,
$$\displaystyle -\frac{1}{2}  \sum_{k=1}^N   \sum_{\text{all} \, \vect{s}}  \pzk^{(c)}
\left[
(\vect{y}_k - {\mathcal{I}}_{\vect{s}} \matr{B}^{*} \phi )^T\matr{V}_{\vect{s}}^{-1} (\vect{y}_k - {\mathcal{I}}_{\vect{s}} \matr{B}^{*} \phi )\right] 
  -    \phi^T {\rm{diag}}( \lambda_1 \matr{R}, \ldots, \lambda_J \matr{R}) \phi. \nonumber $$

This expression is quadratic in $\phi$ and is easily maximized in closed form.  Let $\phi^{(c+1)}$ be this maximizing $\phi$ when we set $\matr{V}=\matr{V}^{(c)}$. So we let $\allf^{(c+1)} = \matr{B}^{*}\phi^{(c+1)}$.

\vspace{.3cm}
\noindent Step 2. \textit{Updating $\matr{V}$.}

For a model with $\matr{V}_{\vect{s}} \equiv \matr{V}$, with no dependence on the state vector $\vect{s}$ and no restrictions on the form of $\matr{V}$,  we show that $\matr{V}^{(c+1)}$ is 
\begin{equation}
\widehat{\matr{V}} = \frac{1}{N} \sum_{k=1}^{N} \sum_{\text{all} \, \vect{s}}  \pzk^{(c)}\big(\vect{y}_k - f_{\vect{s}}(\vect{x}) \big) \big(\vect{y}_k - f_{\vect{s}}(\vect{x})\big)^T
\nonumber
\label{eq:Vhat}
\end{equation}

\noindent with $f_{\vect{s}}(\vect{x})=f_{\vect{s}}(\vect{x})^{(c+1)}$. Note that if the values of $\vect{z}_k$ were non-random and known, then $\pzk^{(c)}$ is a delta function and so 
$\widehat{\matr{V}}$ is similar to the sample covariance matrix of the $\vect{y}_k$s.

When  $\matr{V}_{\vect{s}} \equiv \matr{V}$ follows a homogeneous random intercept model we update the parameter estimates of the restricted $\matr{V}$ in (\ref{restricted_V}) as follows.
Let $\sigma^{2\,(c+1)}$ be 
\begin{eqnarray}
\hat{\sigma}^2 &=& \frac{1}{N(n - 1)}\left(  \sum_{k=1}^{N} \sum_{\text{all} \, \vect{s}}  \pzk^{(c)} (\vect{y}_k - f_{\vect{s}}(\vect{x}) )^T (\vect{y}_k - f_{\vect{s}}(\vect{x})) \right. 
\nonumber \\
&& \left.- \, \frac{1}{n} \sum_{k=1}^{N} \sum_{\text{all} \, \vect{s}}  \pzk^{(c)}
\big[\big(\vect{y}_k - f_{\vect{s}}(\vect{x})\big)^T\vect{1}\big]^2  \right), \nonumber
\label{eq:sigmahat.RI}
\end{eqnarray}

\noindent and $d^{(c+1)}$ be 
\begin{equation}
\hat{d} =  \frac{1}{\sigma^2 N n^2}\sum_{k=1}^{N} \sum_{\text{all} \, \vect{s}}  \pzk^{(c)}
\big[\big(\vect{y}_k - f_{\vect{s}}(\vect{x})\big)^T\vect{1}\big]^2   - \frac{1}{n} \nonumber
\label{dhat}
\end{equation}
\noindent with $\sigma^2$ replaced by $\sigma^{2\,(c+1)}$.  Therefore  $\tau^{2\, (c+1)}= d^{(c+1)}\times \sigma^{2\,(c+1)}$.  
\vspace{.2cm}

The maximization in Step 2 when $\matr{V}_{\vect{s}}$ follows the non-homogeneous random intercept model is given in Section 1.2 of the Supplementary Material, for the case of $J=2$ states. The ECM algorithm for diagonal $\matr{V}_{\vect{s}}$  is given in Section \ref{V:restricted:1}. 

\vspace{.3cm}
\noindent Step 3. \textit{Updating $\alpha$ (any $\matr{V}$)}.  

We maximize  (\ref{S:rep:3}) with respect to the parameters in $\alpha$, with the calculations depending on the proposed model for the hidden states. 

When $z_{1k},\ldots,z_{nk}$ are \textit{iid} with $p_j=p(z_{ik}=j | \alpha)$ we obtain
\begin{equation}
p_j^{(c+1)}= \frac{1}{Nn} \displaystyle \sum_{k=1}^N \sum_{\text{all} \, \vect{s}} \pzk^{(c)} n_{\matr{s},j}.
\nonumber
\end{equation}

For Markov $z_{ik}$s, where the vector $\alpha$ is composed of transition probabilities $a_{lj}$ and initial probabilities $\pi_j$, we obtain
\begin{equation}
\pi_j^{(c+1)} =\frac{1}{N}\displaystyle\sum_{k=1}^N \sum_{\text{all} \, \vect{s}} \pzk^{(c)}  \I(s_1=j) \nonumber
\end{equation}
\noindent and
\begin{equation}
a_{lj}^{(c+1)}=  \frac{ \displaystyle\sum_{k=1}^N \sum_{\text{all} \, \vect{s}} \pzk^{(c)} n_{\vect{s},lj}}
{ \displaystyle\sum_{k=1}^N \sum_{\text{all} \, \vect{s}} \pzk^{(c)} \sum_{i=2}^{n} \I(s_{i-1}=l)},\nonumber
\end{equation}

\noindent where $n_{\vect{s},lj}$ is the number of transitions in $\vect{s}$ from state $l$ to  state $j$, that is, $n_{\vect{s},lj} = \sum_{i=2}^n  \I \{ s_{i-1}=l, s_i=j\}$. 

When $z_{1k},\ldots,z_{nk}$  are independent with the distribution of $z_{ik}$ depending on some covariate(s), $\alpha$ contains the regression coefficients from our logistic regression model for $p(z_{ik} = j | \vect{v}_{ik}, \alpha) \equiv p_j(\vect{v}_{ik}, \alpha)$. In this case, (\ref{S:rep:3}) becomes
\begin{equation}
 \sum_{k=1}^N  \sum_{\text{all} \; \vect{s}} 
  \pzk^{(c)}  ~
   \sum_{i=1}^n \sum_{j=1}^J 
 \log   p_j(\vect{v}_{ik},\alpha)  ~ {\rm{I}} \{ s_{i} = j \}, \nonumber 
\end{equation}
which must be maximized numerically, for instance, via a Newton-Raphson method. 

\subsection{ECM algorithm when $\matr{V}$ is diagonal   \label{V:restricted:1} }

Recall that we consider two cases of $\matr{V}$ diagonal, one with  $\matr{V}=\sigma^2 \matr{I} $ and one with  $\matr{V} = \matr{V}_{\vect{z}_k}=\mbox{diag}(\sigma^2_{z_{1k}},\ldots,\sigma^2_{\vect{z}_{nk}})$.  
We could use the notation and steps of Section \ref{Mstep:rep:iid}, modifying Step 2 for these types of $\matr{V}$.  However, it is much easier to re-derive all three steps using the independence of the components of $\vect{y}_k$ in order to rewrite $\el_1(\gamma)$, and thus  $S(\theta,\theta^{(c)})$, in simpler form.  We will see below that, instead of the  $\pzk^{(c)}$s in (\ref{S:rep:1}) and (\ref{S:rep:3}), we require the simpler
$$p_{ik}(j)^{(c)}= p(z_{ik}=j|\vect{y}_{k},\theta^{(c)}).$$
\noindent The forms of $p_{ik}(j)^{(c)}$ are given in Section 1.3 of the Supplementary Material. 

Here, we carry out the calculations of the ECM algorithm for the case that $ \matr{V}_{\vect{z}_k}=\mbox{diag}(\sigma^2_{z_{1k}},\ldots,\sigma^2_{\vect{z}_{nk}})$, as they can be easily modified for the case that $ \matr{V}=\sigma^2 \matr{I} $: simply replace $\sigma^2_j$ by $\sigma^2$.

We want to find $\theta=\{ f_j(\vect{x}), \sigma_j^2, \, j=1,\ldots,J,$ and $\alpha \}$ that maximizes
\begin{eqnarray}
S^*(\theta,\theta^{(c)})  &=&
 - \displaystyle \frac{1}{2} \sum_{k=1}^{N} \sum_{i=1}^n \sum_{j=1}^J p_{1k}(j)^{(c)}  \log \sigma_j^2 
\label{S:I:0}
\\
&&  - \displaystyle \frac{1}{2} \sum_{k=1}^{N} \sum_{j=1}^J
  \big(\vect{y}_k - f_{j}(\vect{x}) \big)^T \matr{W}_{kj} \big(\vect{y}_k - f_{j}(\vect{x})\big)  \label{S:I:1} 
\\
&&  + ~~ P(f_1,\ldots,f_J,\lambda_1,\ldots,\lambda_J) \label{S:I:2} \\
&&  + ~~ \E_{\theta^{(c)}}(\el_2(\alpha)|\yrep),
\label{S:I:3}
\end{eqnarray}

\noindent where
\begin{equation}
\matr{W}_{kj} = \sigma_j^{-2}\mbox{diag}(p_{1k}(j)^{(c)},\ldots, p_{nk}(j)^{(c)} ).
\label{Wkjdef}
\end{equation}

\noindent We apply the ECM algorithm as follows.

\begin{enumerate}
\item \textit{Updating} the $f_j(\vect{x})$s. Holding the $\sigma_j^2$s and the parameters in $\alpha$ fixed and maximizing the sum of (\ref{S:I:1}) and (\ref{S:I:2}) with respect to $f_j(\vect{x})$ we obtain
\begin{equation}
\hat{f}_{j}(\vect{x}) = \sum_{k=1}^N \matr{H}_{kj}(\lambda_j) \vect{y}_k, \nonumber
\label{fhatVsigma2}
\end{equation}
\noindent where 
\begin{equation}
\matr{H}_{kj}(\lambda) = \vect{B}\left( \matr{B}^T \sum_{r=1}^N \matr{W}_{rj}\, \matr{B} + 2\lambda\matr{R} \right)^{-1}\matr{B}^T \matr{W}_{kj}.
\label{Hhat_Vdiag}
\end{equation}

We let $f_j(\vect{x})^{(c+1)}$ be $\hat{f}_j(\vect{x})$ with $\sigma_j^2$ in $\matr{W}_{kj}$ replaced by $\sigma_j^{2(c)}$.

\item \textit{Updating the $\sigma_j^2$s}. Holding the $f_j(\vect{x})$s and $\alpha$ fixed and maximizing the sum of (\ref{S:I:0}) and  (\ref{S:I:1}) with respect to $\sigma_j^2$ we get

$$\hat{\sigma}_j^2=\frac{\displaystyle  \sum_{k=1}^N\sum_{i=1}^n p_{ik}(j)^{(c)}\big[y_{ik}-f_j(x_i)\big]^2}{\displaystyle \sum_{k=1}^N\sum_{i=1}^n p_{ik}(j)^{(c)}}.$$

Let $\sigma_j^{2(c+1)}$ be $\hat{\sigma}_j^2$ with $f_j(x_i)=f_j(x_i)^{(c+1)}$. 

\item \textit{Updating $\alpha$}.  Hold the $f_j(\vect{x})$s and the $\sigma_j^2$s fixed and maximize (\ref{S:I:3})  with respect to the parameters in $\alpha$.  For \textit{iid} $z_{ik}$s  we obtain
$$p_j^{(c+1)}=  \frac{1}{Nn} \displaystyle \sum^{N}_{k=1} \sum_{i=1}^n p_{ik}(j)^{(c)}.$$
\noindent For Markov $z_{ik}$s, we have 
    $$a_{lj}^{(c+1)}=  \frac{\displaystyle \sum_{k=1}^N\sum_{i=2}^n p(z_{(i-1)k}=l, z_{ik}=j|\mathbf{y}_k,\theta^{(c)})}{\displaystyle \sum_{k=1}^N\sum_{i=2}^n p(z_{(i-1)k}=l|\mathbf{y}_k,\theta^{(c)})}$$

    \noindent and $$\pi_j^{(c+1)} = \frac{1}{N}\sum_{k=1}^N p_{1k}(j)^{(c)}.$$
  \noindent For $z_{ik}$s independent with distribution of $z_{ik}$ depending on some covariates we need numerical optimization methods, such as Newton-Raphson, to obtain the coefficient estimates from our logistic regression model for $p(z_{ik} = j | \vect{v}_{ik}, \alpha)$. So, for example, if there are $J=2$ states and the covariate vector is $\vect{v}_{ik} = (1,v_{ik})^T$, we apply a numerical method to obtain $\beta_{20}$ and $\beta_{21}$ that maximize
 $$ 
\E_{\theta^{(c)}}(\el_2(\beta_{20},\beta_{21})|\yrep)  = 
 \sum_{k=1}^N  \sum_{i=1} ^n  \big\{p_{ik}(2)^{(c)} (\beta_{20} + \beta_{21}v_{ik}) - \log(1+e^{\beta_{20} + \beta_{21}v_{ik}}) \big\}.$$
\end{enumerate}

\subsection{Choice of the smoothing parameters when $\matr{V}$ is diagonal}  \label{choice:lambda:rep}

In principal, we can always compute the smoothing parameters by ``leave-one-curve-out" cross-validation.   However, for many models, this can be computationally intensive.  Fortunately, in the models with $\matr{V}=\sigma^2 \matr{I}$ or $ \matr{V}_{\vect{z}_k}=\mbox{diag}(\sigma^2_{z_{1k}},\ldots,$       
$\sigma^2_{z_{nk}})$, we can shorten calculations by using Theorem 1 below. In this section, we describe our iterative cross-validation procedure,
 implemented for our data analysis in Section 
\ref{uncorr:obs} for $ \matr{V}_{\vect{z}_k}=\mbox{diag}(\sigma^2_{z_{1k}},\ldots,$       
$\sigma^2_{z_{nk}})$.  The steps for $\matr{V}=\sigma^2 \matr{I}$
are the same  except with  $\hat{\sigma}^2$ replacing the $\hat{\sigma}_j^2$s . 

In our data analysis we set the initial values, the $\lambda_j^{(0)}$s, to those that worked well when tested on the data set.
We update the $\lambda_j$s  as follows.

\begin{enumerate}
\item  At iteration $i$, with $\lambda_j=\lambda_j^{(i)}$, $j=1,\ldots,J$, use the ECM algorithm of Section \ref{V:restricted:1} to find the $\hat{p}_{ik}(j)$s, the $\hat{\sigma}_j^2$s and the $\hat{f}_j$s.
    \item Discard the $\hat{f}_j$s from Step 1.
\item  
Let $\widehat{\matr{W}}_{kj}$ be  $\matr{W}_{kj}$ as defined in (\ref{Wkjdef}) but with the $\hat{\sigma}_j^2$s and $\hat{p}_{ik}(j)$s replacing the
${\sigma}_j^2$s and ${p}_{ik}(j)$s.   Treat the $\hat{\sigma}_j^2$s and the $\hat{p}_{ik}(j)$s and thus the $\widehat{\matr{W}}_{kj}$s  as fixed. 
\item 
For $j=1,\ldots, J$, over a grid of possible $\lambda$ values, set $\lambda_j^{(i+1)}$ as the value of $\lambda$ that minimizes the following leave-one-replicate-out cross-validation criterion:
\begin{equation}
CV_j(\lambda) = \sum_{k=1}^N \big[ \vect{y}_k -\hat{f}^{(-k)}_{j\,\lambda} (\vect{x})  \big]^T \widehat{ \matr{W}}_{kj} \big[ \vect{y}_k -\hat{f}^{(-k)}_{j\,\lambda} (\vect{x} ) \big]
\label{CVrep}
\end{equation}
\noindent  where $\hat{f}^{(-k)}_{j\,\lambda} $ is the function that maximizes
\begin{equation}
S_j^{(-k)}(f_j) = -\frac{1}{2}\sum_{r=1:r\neq k}^N  \big[ \vect{y}_r -f_j(\vect{x})\big]^T \widehat{\matr{W}}_{rj} \big[ \vect{y}_r -f_j(\vect{x})\big] +
P(f_j,\lambda). \nonumber
\label{S_no_k}
\end{equation}
\item
Repeat steps 1-4 with $\lambda_j = \lambda_j^{(i+1)}$, $j=1,\ldots,J$, until convergence.
\end{enumerate}

\noindent We use the final values of the $\lambda_j$s to obtain all of the parameter estimates from the ECM algorithm as in
Section {\ref{V:restricted:1}}.

Finding $\lambda$ that minimizes (\ref{CVrep}) is computationally intensive. Fortunately, we have the following theorem.

\newtheorem{thm}{Theorem}
\begin{thm}
Let  $\hat{f}_{1\lambda}$,  $\ldots, \hat{f}_{J\lambda}$
be the maximizers of the sum of (\ref{S:I:1}) and (\ref{S:I:2}),  with $\matr{W}_{kj}$ replaced by $\widehat{\matr{W}}_{kj}$.  Let $\widehat{\matr{H}}_{kj}$ be as in  (\ref{Hhat_Vdiag}), but  with $\matr{W}_{kj}$ replaced by $\widehat{\matr{W}}_{kj}$.  
Suppose that $\matr{I} - \widehat{\matr{H}}_{kj}$ is invertible and $\widehat{\matr{W}}_{kj}$ is positive definite, $j=1,\ldots, J$. Then 
\begin{equation}
CV_j(\lambda) = \sum_{k=1}^N \Big[ (\matr{I}  - \widehat{\matr{H}}_{kj}(\lambda))^{-1} ( \hat{f}_{j\,\lambda}(\vect{x}) - \vect{y}_k)   \Big]^T \widehat{\matr{W}}_{kj} \Big[ (\matr{I}  - \widehat{\matr{H}}_{kj}(\lambda))^{-1} ( \hat{f}_{j\,\lambda}(\vect{x}) - \vect{y}_k) \Big]. \nonumber
\end{equation}
\label{thm1}
\end{thm}

\noindent The proof follows directly from Lemma 2 in the Appendix, which holds in a slightly more general setting.

\section{Simulation study}\label{sec:sim_studies}

We carry out a simulation study under three different designs considering that the hidden states, the $z_{ik}$s, can take values 1 or 2.  For each design 300 independent data sets are generated, each with $N=100$ replicates.  In design 1,  $z_{1k}, \ldots, z_{nk}$ are \textit{iid} and
$\matr{V}$ follows the homogeneous random intercept model as in (\ref{restricted_V}).  In design 2,  $z_{1k}, \ldots, z_{nk}$ follow a Markov structure
and $\matr{V}$ also follows the homogeneous random intercept model.  In design 3, $z_{1k}, \ldots, z_{nk}$ are independent with the distribution of $z_{ik}$ depending on a univariate covariate,  $v_{ik}$.  In this third design, we take $\matr{V} = \sigma^2 \matr{I}$. 
To study all three designs we use the same vector of evaluation points $\vect{x}$ and the same true functions $f_1$ and $f_2$. The vector $\vect{x}=(x_1,\ldots,x_n)^T$ consists of $n=10$ equally spaced points, $1, 12, 23,  \ldots, 89, 100$. The true function $f_2$ is the same we used in the simulation study presented in \citet{deSouza2014}. The true function $f_1$ is simply $f_2 - 0.1$. In the third study, for each simulated data set,  we generate   $v_{ik}, k=1,\ldots,n, i=1,\ldots,N$.  Figure \ref{Sim:fitted:example} contains example data sets generated from each of the three designs.

For Designs 1 and 2 we generate each simulated data set as follows.
\begin{enumerate}
\item Generate the $z_{ik}$s according to the specified model -  \textit{iid} for Design 1, Markov for Design 2.  For the \textit{iid} model, we set $p_1=p(z_{ik}=1)=0.5$.  For Markov $z_{ik}$s, we set transition probabilities $a_{12}=p(z_i=2|z_{i-1}=1)=0.3$ and $a_{21}=p(z_i=1|z_{i-1}=2)=0.4$ and initial probabilities $\pi_1=\pi_2=0.5$. 
\item Generate the $y_{ik}$s according to the homogeneous random intercept model of Section  \ref{Sec:choices_alpha_V} with
   $\tau^2 = 10^{-4}$ and $\sigma^2=10^{-5}$.
\item Repeat steps 1 and 2 $N=100$ times to obtain a data set of 100 replicates.
\end{enumerate}

For Design 3 we generate each simulated data set as follows.

\begin{enumerate}
\item Generate $v_{ik}$s \textit{iid} $N(0,1)$.
\item Generate the $z_{ik}$s such that  $p(z_{ik}=1|v_{ik}) \equiv p_1(v_{ik}) = 1/  [ 1 + \exp( \beta_{0} + \beta_{1}v_{ik})]$
 and so $\log [ p_2(v_{ik})/p_1(v_{ik}) ] =\beta_{0} + \beta_{1}v_{ik}$.  We set $\beta_0=2$ and $\beta_1=5$.
\item Generate the $y_{ik}$s as follows.
  If $z_{ik} = 1$ then  $y_{ik} = f_1(x_i) + e_{ik}$. 
  If $z_{ik} = 2$ then $y_{ik} = f_2(x_i) + e_{ik}$.   The
 $\epsilon_{ik}$s are \textit{iid} $N(0,\sigma^2)$. We set $\sigma^2=5\times10^{-5}$.
\item Repeat steps 1, 2 and 3 $N=100$ times to obtain a data set of 100 replicates.
\end{enumerate}

We analyze the simulated data under each design using the proposed EM algorithm. We set initial parameter values to the true parameter values to speed up computation. We did try initial values that were different than the true parameter values and the EM algorithm also converged, but it took longer than when starting from the truth, as expected.

The values of $\lambda_1$ and $\lambda_2$ are fixed and equal to $10^{-4}$ in the study of all designs. We choose this value by examining a few simulated data sets and a range of lambda values.  We find that the results of these preliminary analyses are not sensitive to the choice of smoothing parameter over a wide range of lambda values.

\subsection{Results}\label{results:rep}

The three plots in Figure \ref{Sim:fitted:example} show the fitted values $\hat{f}_1(\vect{x})$ and $\hat{f}_2(\vect{x})$ (dashed curves) for a data set generated from each simulation design.

We assess the quality of the estimated functions via the pointwise empirical mean squared error (EMSE) as in \citet{deSouza2014}. For all designs $\hat{f}_1(\vect{x})$ and $\hat{f}_2(\vect{x})$ produce very small values of EMSE ($< 2\times 10^{-6}$).  However, when generating data according to Design 3, the EMSE values for $\hat{f}_1(\vect{x})$  are larger than for Designs 1 and 2. 
  
We observe that in all cases we are slightly underestimating the values of the variance parameters. This may be due to the challenges of correctly adjusting the degrees of freedom in the estimates,  in order to account for the estimation of the $f_j$s.  Recall that, in Designs 1 and 2, the error variance satisfies $10^{5}\sigma^2=1$ and in Design 3, $10^{5}\sigma^2=5$. The averages of our estimates of $10^{5} \sigma^2$ (with standard errors) under Designs 1, 2 and 3 are, respectively, 0.978 (0.046), 0.977 (0.045) and 4.919 (0.238). In Designs 1 and 2, we have an additional variance parameter, namely, the variance of the random effect intercept, with $10^{4}\tau^2=1$. In these cases, the averages of $10^{4}\tau^2$ are equal to 0.977 with standard deviations equal to 0.152.

Table \ref{tb:zparam:rep} contains the mean and the standard deviation of the estimates of the parameters of the latent process under each simulation design, along with the averages of our proposed standard errors (SEs).  Note that the standard deviations of the estimates are close to the values of the means of the proposed SEs, as desired. Table \ref{tb:zparam:rep} also shows the empirical coverage percentages of a 90\% and a 95\% confidence interval. We consider confidence intervals of the form
``mean of the parameter estimates  $\pm z_{\alpha/2} ~ \times$  proposed SE", 
 where $z_{\alpha/2}$ is the $\alpha/2$ quantile of a standard normal distribution with $\alpha=0.1$ and 0.05. The empirical coverage percentages under all three simulation designs are very close to the true level of the corresponding confidence interval.

\section{Analysis of  the power usage data}\label{data_analysis}

The data shown in Figure \ref{building_replicates} consist of daytime hourly power usage of a building from 9am to 4pm ($n=8$ observations in a day) on $N=44$ business days in June and July 2009. For the same days and hours we also have available the temperature at a local weather station. We apply our proposed methodology to these  data treating each day as a replicate and modelling power usage as arising from $J=2$ functions, one function giving power usage when the chiller is off ($j=1$), and the other function giving power usage when the chiller is on ($j=2$). 
In Section \ref{uncorr:obs} we present the results assuming the covariance matrix $\matr{V}$ is diagonal
and in Section \ref{corr:obs} we present the results when we assume  $\matr{V}$ is generated by the non-homogeneous random intercept model  as in (\ref{covEq:RI:diff}).

\subsection{Results: diagonal $\matr{V}$}\label{uncorr:obs}

In this section we consider two models for $\matr{V}$: $\matr{V}=\sigma^2\matr{I}$ and $\matr{V}= \matr{V}_{\vect{z}_k}= \mbox{diag}(\sigma^2_{z_{1k}},\ldots,\sigma^2_{z_{nk}})$. We use the ECM algorithm described in Section \ref{V:restricted:1} to estimate the model parameters considering \textit{iid} $z_{ik}$s, Markov $z_{ik}$s and $z_{ik}$s that are independent with distribution depending on temperature. The smoothing parameters, the $\lambda_j$s, are chosen by cross-validation as described in detail in Section \ref{choice:lambda:rep}.

Figures \ref{building_noHMM_equals2} and \ref{building_noHMM_diffs2} present the fitted functions for \textit{iid} hidden states $z_{ik}$s when we assume $\matr{V}=\sigma^2\matr{I}$ and  $\matr{V}_{\vect{z}_k}$, respectively. We can observe that the fitted curves are very similar in the two figures. The estimated curve giving power usage when the chiller is on, obtained assuming $\matr{V}=\sigma^2\matr{I}$, is slightly smoother than the one obtained assuming $\matr{V}_{\vect{z}_k}$.
Table \ref{tb:Building:iid}  presents the parameter estimates and chosen $\lambda_j$s. We can see that the estimates of $p_j=p(z_{ik}=j)$ from the two models for $\matr{V}$ agree within the reported standard errors. We also observe in the lower half of the Table that the estimated variance when the chiller is on is much higher than when the chiller is off.

Figures \ref{building_HMM_equals2} and \ref{building_HMM_diffs2} present the fitted curves for Markov $z_{ik}$s when we assume $\matr{V}=\sigma^2\matr{I}$ and  $\matr{V}_{\vect{z}_k}$, respectively. As in the \textit{iid} case, the fitted curve giving power usage when the chiller is on obtained assuming $\matr{V}=\sigma^2\matr{I}$ is slightly smoother than the one obtained assuming $\matr{V}_{\vect{z}_k}$.
Table \ref{tb:Building:Markov}  provides information on the estimated model parameters and the chosen smoothing parameters.   As in the $iid$ case, the estimated variance when the chiller is on is much higher than when the chiller is off. We observe that  the estimates of $a_{21}$, the transition probability from ``chiller on" to ``chiller off", are very small or equal to zero.  Any estimate of $a_{21}$ is  expected to be small, as   there is only one replicate in the data set where we observe this transition. The estimate of zero is reasonable when we assume different variances;    $\hat{a}_{21}$ is zero because the transition happens gradually, which our model does not allow, and the method incorrectly classifies all observations as coming from the condition ``chiller on", failing to detect the transition.   This replicate is  the green curve in Figure \ref{building_HMM_diffs2}.   

Figure \ref{building_temperature} presents the fitted curves when we assume the $z_{ik}$s are independent with the distribution of $z_{ik}$ depending on temperature via the following logistic regression model: 
$$\log \frac{ p(\mbox{chiller on}~ | ~\mbox{temperature} )} { p(\mbox{chiller off} ~ | ~\mbox{temperature} )} = \beta_{0} + \beta_{1}~\mbox{temperature}.$$ 

\noindent Table \ref{tb:Building:cov}
shows the corresponding estimated model parameters assuming $\matr{V}_{\vect{z}_k}$ along with the chosen smoothing parameters, the $\lambda_j$s. We observe in Table \ref{tb:Building:cov} that the standard error for $\hat{\beta}_{1}$ is very small and by considering a confidence interval of the form $\hat{\beta}_{1}  \pm 1.96 \times  \mbox{SE}(\hat{\beta}_{1})$ we conclude that the coefficient $\beta_{1}$ is statistically significant. 

\subsection{Results: correlated observations generated by the non-homogeneous random intercept model} \label{corr:obs}

In the analyses of Section \ref{uncorr:obs}, we see that the variability in energy consumption when the chiller is on is higher than when the chiller is off.  Thus, models such as $\matr{V} = \sigma^2 \matr{I}$ or $\matr{V}$ following the homogeneous random intercept model may not be appropriate.  Therefore, to model this heterogeneity in variance and the correlation between observations, 
we fit the proposed switching nonparametric regression model to the power usage data assuming the covariance matrix $\matr{V}$ is generated by the non-homogeneous random intercept model as in (\ref{covEq:RI:diff}). We use the ECM algorithm described in Section \ref{sec:parameter_estimation}  and in Section 1.2 of the Supplementary Material to obtain the parameter estimates. We conduct the analysis assuming the hidden states $z_{ik}$s are \textit{iid}. We assume that $\lambda_1=\lambda_2=\lambda$ and choose the smoothing parameters via a ``brute force" leave-one-curve-out method over a grid of possible values of $\lambda$ (see Table 1 of Supplementary Material).  

Table \ref{tb:Building:iid:RI:complex} presents the parameter estimates. We observe that the estimates of $p_1$ and $p_2$ in Table \ref{tb:Building:iid:RI:complex} agree within the reported standard errors with the estimates obtained in Table \ref{tb:Building:iid}  where we assume the observations are uncorrelated. Figure \ref{building_RI_noHMM_complex} shows the corresponding fitted curves. We can observe that the fitted function corresponding to the condition ``chiller on" is lower than that in  Figures  \ref{building_noHMM_equals2} to \ref{building_temperature}.  The non-homogeneous random intercept model appears to ``explain" days of high power usage by a larger variability of the ``chiller on" random intercept.    Thus the replicates with very high power usage have less of an impact on the final fitted  ``chiller on"  curve.

\section{Discussion}\label{sec:discussion}

We have introduced a method for the analysis of data arising from random samples of a process with a complex structure.  The structure  depends on  a latent state process where each state corresponds to a true  smooth regression function.   The estimation techniques and standard error calculations were developed for  several specific cases of state processes and error covariances.      
We have considered restrictive covariance structures, save for the case where $\matr{V}$ is completely unrestricted.  While the covariance models we consider may not capture all of the dependencies in a data set, our techniques and ideas should carry over to more complex time series modelling of the error process. For instance, we can model more complicated covariance structures via random regression approaches, such as with B-spline basis functions or with lines that have random slopes in addition to random intercepts. Similarly, we can use our methods to consider more complex models for the latent process, such as a Markov model with covariate-dependent transition probabilities.
Further useful extensions might incorporate a dependence among replicates; for instance, in studying energy consumption of several buildings, one would want to incorporate a random ``building" effect.

\section*{Acknowledgements}
We would like to thank the Editor, Associate Editor and reviewers for their insightful questions and comments. 

\bibliographystyle{imsart-nameyear}
\bibliography{references}

\begin{thebibliography}{17}
% BibTex style file: imsart-nameyear.bst, 2010-01-14
% Default style options (sort=1,type=nameyear).
% Used options (sort=1,type=nameyear).

\bibitem[\protect\citeauthoryear{Brown, Barrington-Leigh and
  Brown}{2012}]{brown2012kernel}
\begin{barticle}[author]
\bauthor{\bsnm{Brown},~\bfnm{M.}\binits{M.}},
  \bauthor{\bsnm{Barrington-Leigh},~\bfnm{C.}\binits{C.}} \AND
  \bauthor{\bsnm{Brown},~\bfnm{Z.}\binits{Z.}}
(\byear{2012}).
\btitle{Kernel regression for real-time building energy analysis}.
\bjournal{Journal of Building Performance Simulation}
\bvolume{5}
\bpages{263--276}.
\end{barticle}
\endbibitem

\bibitem[\protect\citeauthoryear{Capp{\'e}, Moulines and
  Ryd{\'e}n}{2005}]{cappe2005}
\begin{bbook}[author]
\bauthor{\bsnm{Capp{\'e}},~\bfnm{O.}\binits{O.}},
  \bauthor{\bsnm{Moulines},~\bfnm{E.}\binits{E.}} \AND
  \bauthor{\bsnm{Ryd{\'e}n},~\bfnm{T.}\binits{T.}}
(\byear{2005}).
\btitle{Inference in Hidden Markov Models}.
\bpublisher{Springer Verlag}.
\end{bbook}
\endbibitem

\bibitem[\protect\citeauthoryear{De~Souza and Heckman}{2014}]{deSouza2014}
\begin{barticle}[author]
\bauthor{\bsnm{De~Souza},~\bfnm{C.~P.~E.}\binits{C.~P.~E.}} \AND
  \bauthor{\bsnm{Heckman},~\bfnm{N.~E.}\binits{N.~E.}}
(\byear{2014}).
\btitle{Switching nonparametric regression models}.
\bjournal{Journal of Nonparametric Statistics}
\bvolume{26}
\bpages{617--637}.
\end{barticle}
\endbibitem

\bibitem[\protect\citeauthoryear{Dempster, Laird and
  Rubin}{1977}]{dempster1977em}
\begin{barticle}[author]
\bauthor{\bsnm{Dempster},~\bfnm{Arthur~P}\binits{A.~P.}},
  \bauthor{\bsnm{Laird},~\bfnm{Nan~M}\binits{N.~M.}} \AND
  \bauthor{\bsnm{Rubin},~\bfnm{Donald~B}\binits{D.~B.}}
(\byear{1977}).
\btitle{Maximum likelihood from incomplete data via the EM algorithm}.
\bjournal{Journal of the Royal Statistical Society Series B}
\bvolume{39}
\bpages{1--38}.
\end{barticle}
\endbibitem

\bibitem[\protect\citeauthoryear{James, Hastie and
  Sugar}{2000}]{james2000principal}
\begin{barticle}[author]
\bauthor{\bsnm{James},~\bfnm{Gareth~M}\binits{G.~M.}},
  \bauthor{\bsnm{Hastie},~\bfnm{Trevor~J}\binits{T.~J.}} \AND
  \bauthor{\bsnm{Sugar},~\bfnm{Catherine~A}\binits{C.~A.}}
(\byear{2000}).
\btitle{Principal component models for sparse functional data}.
\bjournal{Biometrika}
\bvolume{87}
\bpages{587--602}.
\end{barticle}
\endbibitem

\bibitem[\protect\citeauthoryear{Langrock et~al.}{2017}]{langrock2014markov}
\begin{barticle}[author]
\bauthor{\bsnm{Langrock},~\bfnm{Roland}\binits{R.}},
  \bauthor{\bsnm{Kneib},~\bfnm{Thomas}\binits{T.}},
  \bauthor{\bsnm{Glennie},~\bfnm{Richard}\binits{R.}} \AND
  \bauthor{\bsnm{Michelot},~\bfnm{Th{\'e}o}\binits{T.}}
(\byear{2017}).
\btitle{Markov-switching generalized additive models}.
\bjournal{Statistics and Computing}
\bvolume{27}
\bpages{259--270}.
\bdoi{10.1007/s11222-015-9620-3}
\end{barticle}
\endbibitem

\bibitem[\protect\citeauthoryear{Louis}{1982}]{louis1982}
\begin{barticle}[author]
\bauthor{\bsnm{Louis},~\bfnm{T.~A.}\binits{T.~A.}}
(\byear{1982}).
\btitle{Finding the observed information matrix when using the EM algorithm}.
\bjournal{Journal of the Royal Statistical Society Series B}
\bvolume{44}
\bpages{226--233}.
\end{barticle}
\endbibitem

\bibitem[\protect\citeauthoryear{MacDonald}{2014}]{macdonald2014numerical}
\begin{barticle}[author]
\bauthor{\bsnm{MacDonald},~\bfnm{Iain~L}\binits{I.~L.}}
(\byear{2014}).
\btitle{Numerical Maximisation of Likelihood: A Neglected Alternative to EM?}
\bjournal{International Statistical Review}
\bvolume{82}
\bpages{296--308}.
\end{barticle}
\endbibitem

\bibitem[\protect\citeauthoryear{McLachlan and Krishnan}{2008}]{mclachlan}
\begin{bbook}[author]
\bauthor{\bsnm{McLachlan},~\bfnm{G.~J.}\binits{G.~J.}} \AND
  \bauthor{\bsnm{Krishnan},~\bfnm{T.}\binits{T.}}
(\byear{2008}).
\btitle{{The EM Algorithm and Extensions}}.
\bpublisher{2nd Ed., Wiley New York}.
\end{bbook}
\endbibitem

\bibitem[\protect\citeauthoryear{Meng and Rubin}{1993}]{meng}
\begin{barticle}[author]
\bauthor{\bsnm{Meng},~\bfnm{X.~L.}\binits{X.~L.}} \AND
  \bauthor{\bsnm{Rubin},~\bfnm{D.~B.}\binits{D.~B.}}
(\byear{1993}).
\btitle{{Maximum likelihood estimation via the ECM algorithm: a general
  framework}}.
\bjournal{Biometrika}
\bvolume{80}
\bpages{267-278}.
\end{barticle}
\endbibitem

\bibitem[\protect\citeauthoryear{Ou and Martin}{2008}]{ou2008}
\begin{barticle}[author]
\bauthor{\bsnm{Ou},~\bfnm{X.}\binits{X.}} \AND
  \bauthor{\bsnm{Martin},~\bfnm{E.}\binits{E.}}
(\byear{2008}).
\btitle{Batch process modelling with mixtures of Gaussian processes}.
\bjournal{Neural {C}omputing \& {A}pplications}
\bvolume{17}
\bpages{471--479}.
\end{barticle}
\endbibitem

\bibitem[\protect\citeauthoryear{Ramsay and Silverman}{2005}]{ramsay2005}
\begin{bbook}[author]
\bauthor{\bsnm{Ramsay},~\bfnm{J.~O.}\binits{J.~O.}} \AND
  \bauthor{\bsnm{Silverman},~\bfnm{BW}\binits{B.}}
(\byear{2005}).
\btitle{Functional Data Analysis}.
\bpublisher{Springer Verlag}.
\end{bbook}
\endbibitem

\bibitem[\protect\citeauthoryear{Rasmussen and
  Ghahramani}{2002}]{rasmussen2002}
\begin{binproceedings}[author]
\bauthor{\bsnm{Rasmussen},~\bfnm{C.~E.}\binits{C.~E.}} \AND
  \bauthor{\bsnm{Ghahramani},~\bfnm{Z.}\binits{Z.}}
(\byear{2002}).
\btitle{Infinite mixtures of Gaussian process experts}.
In \bbooktitle{Advances in {N}eural {I}nformation {P}rocessing {S}ystems 14:
  {P}roceedings of the 2001 {C}onference}
\bvolume{2}
\bpages{881--888}.
\bpublisher{The MIT Press}.
\end{binproceedings}
\endbibitem

\bibitem[\protect\citeauthoryear{Tresp}{2001}]{tresp}
\begin{binproceedings}[author]
\bauthor{\bsnm{Tresp},~\bfnm{V.}\binits{V.}}
(\byear{2001}).
\btitle{Mixtures of Gaussian processes}.
In \bbooktitle{Advances in {N}eural {I}nformation {P}rocessing {S}ystems 13:
  {P}roceedings of the 2000 {C}onference}
\bpages{654--660}.
\bpublisher{The MIT Press}.
\end{binproceedings}
\endbibitem

\bibitem[\protect\citeauthoryear{Wahba}{1990}]{wahba1990spline}
\begin{bbook}[author]
\bauthor{\bsnm{Wahba},~\bfnm{Grace}\binits{G.}}
(\byear{1990}).
\btitle{Spline models for observational data}
\bvolume{59}.
\bpublisher{Siam}.
\end{bbook}
\endbibitem

\bibitem[\protect\citeauthoryear{Yao, M{\"u}ller and
  Wang}{2005}]{yao2005functional}
\begin{barticle}[author]
\bauthor{\bsnm{Yao},~\bfnm{Fang}\binits{F.}},
  \bauthor{\bsnm{M{\"u}ller},~\bfnm{Hans-Georg}\binits{H.-G.}} \AND
  \bauthor{\bsnm{Wang},~\bfnm{Jane-Ling}\binits{J.-L.}}
(\byear{2005}).
\btitle{Functional data analysis for sparse longitudinal data}.
\bjournal{Journal of the American Statistical Association}
\bvolume{100}
\bpages{577--590}.
\end{barticle}
\endbibitem

\bibitem[\protect\citeauthoryear{Zucchini, MacDonald and
  Langrock}{2016}]{zucchini2016hidden}
\begin{bbook}[author]
\bauthor{\bsnm{Zucchini},~\bfnm{Walter}\binits{W.}},
  \bauthor{\bsnm{MacDonald},~\bfnm{Iain~L.}\binits{I.~L.}} \AND
  \bauthor{\bsnm{Langrock},~\bfnm{Roland}\binits{R.}}
(\byear{2016}).
\btitle{Hidden Markov models for time series: an introduction using R, 2nd
  Edition}.
\bpublisher{Chapman and Hall/CRC}.
\end{bbook}
\endbibitem

\end{thebibliography}

\section*{Appendix}
\subsection*{Proof of Theorem 1}  \label{app:proof}

Theorem 1 is based on the following lemmas, which frame the problem for fixed $j$  and fixed $\lambda$ (so these are dropped in notation) and with general matrices $\mathcal{W}_r$, $r=1,\ldots, N$.  Lemma 1 holds for general penalties, while Lemma 2 places further restrictions, restrictions that hold in our setting.  Throughout, we assume that all maximizers exist.

Let $\hat{f}^{(-k)} $ maximize
\[
{S}^{(-k)}(f) = -\frac{1}{2} \sum_{r=1; r \neq k}^N  \big[ \vect{y}_r -f(\vect{x})\big]^T \mathcal{W}_{r} \big[ \vect{y}_r -f(\vect{x})\big] +
P(f).
\]

\newtheorem{lemma}{Lemma}
\begin{lemma}
 Let $\hat{f}^{(\ast k)}$ maximize
\begin{eqnarray}
S^{(\ast k)}(f) &=&
  - \displaystyle \frac{1}{2}  [  \hat{f}^{(-k)}(\vect{x}) - f(\vect{x}) ]^T  \mathcal{W}_k  [  \hat{f}^{(-k)}(\vect{x}) - f(\vect{x}) ]
  \nonumber \\
  &&
  - \displaystyle \frac{1}{2} \sum_{r=1, r \neq k}^{N} 
    [\vect{y}_r - f(\vect{x}) ]^T \mathcal{W}_{r} [ \vect{y}_r - f(\vect{x})]    ~~
 + ~~ P(f).
 \nonumber \end{eqnarray}
If $\mathcal{W}_k$ is positive definite then
%\label{S_k_star}
 $\hat{f}^{(-k)}(\vect{x}) = \hat{f}^{(\ast k)}(\vect{x})$.
\label{lemma1}
\end{lemma}

\subsection*{Proof of Lemma 1.}
For simplicity let $k=1$. We want to show that $\hat{f}^{(-1)} = \hat{f}^{(\ast 1)}$. We know $\hat{f}^{(-1)}$ maximizes $S^{(-1)}(f)$  and, therefore,
\begin{equation}
 S^{(-1)}(\hat{f}^{(-1)}) - S^{(-1)}(\hat{f}^{(\ast 1)}) \geq 0. \nonumber
 \label{proofextraEq}
 \end{equation}

We also know that $\hat{f}^{(\ast 1)}$ maximizes $S^{(\ast 1)}(f)$. Thus, $S^{(\ast 1)}(\hat{f}^{(\ast 1)}) - S^{(\ast 1)}(\hat{f}^{(- 1)}) \geq 0$, that is,
\begin{eqnarray}
&& - ~ \frac{1}{2} \big[ \hat{f}^{(- 1)}(\vect{x}) -\hat{f}^{(\ast 1)}(\vect{x})\big]^T \mathcal{W}_{1} \big[ \hat{f}^{(- 1)}(\vect{x}) -\hat{f}^{(\ast 1)}(\vect{x})\big] 
\nonumber \\
&&
- \frac{1}{2}\sum_{r=2}^N  \big[ \vect{y}_r - \hat{f}^{(\ast 1)}(\vect{x}) \big]^T \mathcal{W}_{r} \big[ \vect{y}_r -\hat{f}^{(\ast 1)}(\vect{x}) \big] 
~ + ~  P(\hat{f}^{(\ast 1)})  
\nonumber
\\
&&
+ ~ \frac{1}{2}\sum_{r=2}^N  \big[ \vect{y}_r - \hat{f}^{(- 1)}(\vect{x}) \big]^T \mathcal{W}_{r} \big[ \vect{y}_r -\hat{f}^{(- 1)} (\vect{x})\big] ~- ~ P(\hat{f}^{(- 1)}) 
\geq 0,
%\label{proofEq:1}
\nonumber
\end{eqnarray}

\noindent such that
\begin{equation}
-\frac{1}{2} \big[ \hat{f}^{(- 1)}(\vect{x}) -\hat{f}^{(\ast 1)}(\vect{x})\big]^T \mathcal{W}_{1} \big[ \hat{f}^{(- 1)}(\vect{x}) -\hat{f}^{(\ast 1)}(\vect{x})\big]  \geq  S^{(- 1)}(\hat{f}^{(- 1)}) - S^{(- 1)}(\hat{f}^{(\ast 1)})  \geq 0, \nonumber
\end{equation}

\noindent which implies that
\begin{equation}
\big[ \hat{f}^{(- 1)}(\vect{x}) -\hat{f}^{(\ast 1)}(\vect{x})\big]^T \mathcal{W}_{1} \big[ \hat{f}^{(- 1)}(\vect{x}) -\hat{f}^{(\ast 1)}(\vect{x})\big] \leq 0, \nonumber
%\label{proofEq:2}
\end{equation} and, because $\mathcal{W}_{1}$ is positive definite,  $\hat{f}^{(- 1)} (\vect{x})= \hat{f}^{(\ast 1)}(\vect{x})$. $\Box$ 

\begin{lemma}
Suppose that $\mathcal{W}_k$ is positive definite for $k=1,\ldots, N$.
Let $\hat{f}$ maximize
\[
S(f) = 
  - \displaystyle \frac{1}{2} \sum_{k=1}^{N} 
    \big[ \vect{y}_k - f(\vect{x}) \big]^T \mathcal{W}_{k} \big[\vect{y}_k - f(\vect{x})\big]    ~~
 + ~~ P(f).
\]
  If there exist matrices $\mathcal{H}_k$, $k=1,\ldots, N$,  not depending on the $\vect{y}_r$s,  such that $\hat{f}(\vect{x})  =  \sum_{k=1}^N  \mathcal{H}_k  \vect{y}_k$,
then 
\[
(\matr{I}  - \mathcal{H}_{k})~ [ \hat{f}^{(-k)}(\vect{x}) - \vect{y}_k ]=  \hat{f}(\vect{x}) - \vect{y}_k. 
\]
\label{lemma2}
\end{lemma}

\subsection*{Proof of Lemma 2}
Note that $\hat{f}^{(\ast k)}$, as defined in Lemma 1,  is the maximizer of $S$ with $\vect{y}_k$ replaced by $\hat{f}^{(-k)}$.
By the assumption of the form of the maximizer of $S$,  $\hat{f}^{(\ast k)}(\vect{x})$ can be written as
\begin{eqnarray}
\hat{f}^{(\ast k)}(\vect{x}) &=& \sum_{r=1: r \neq k}^N \mathcal{H}_{r} \vect{y}_r + \mathcal{H}_{k} \hat{f}^{(-k)}(\vect{x}) \nonumber \\
&=& \hat{f} (\vect{x})-  \mathcal{H}_{k} \vect{y}_k + \mathcal{H}_{k}\hat{f}^{(-k)}(\vect{x}). \nonumber
%\label{fast:sumHk}
\end{eqnarray}

\noindent From Lemma 1 we know $\hat{f}^{(- k)}(\vect{x}) = \hat{f}^{(\ast k)}(\vect{x})$. Thus,
\begin{equation}
\hat{f}^{(-k)}(\vect{x}) =  \hat{f} (\vect{x})-  \mathcal{H}_{k} \vect{y}_k +  \mathcal{H}_{k} \hat{f}^{(-k)}(\vect{x}). 
\nonumber
%\label{fk:sumHk}
\end{equation}

\noindent Now subtracting $\vect{y}_k$ from both sides of this equation, we obtain
\begin{equation}
( \matr{I}  - \mathcal{H}_{k}) [ \hat{f}^{(-k)} (\vect{x})- \vect{y}_k ] =   \hat{f}(\vect{x}) - \vect{y}_k.  \nonumber
%\label{final:identity}
\end{equation}

\noindent $\Box$

\clearpage

\begin{table}
\begin{center}
\caption{\textit{Simulation study.} The mean and the standard deviation (SD) of the estimates of the parameters of the latent state process under each design, along with the mean of our proposed standard errors (SEs) and empirical coverage percentages of the proposed confidence intervals.}
\vspace{-.3cm}
\begin{tabular}{@{}r c c c c c@{}} 
\toprule

& & & & \multicolumn{2}{c}{empirical coverage}   \\ 

Design & true parameters & mean (SD) & mean of SEs &  90\%  & 95\%  \\

\midrule
1 & $p_1=0.5$ & 0.499 (0.016) & 0.016 &  90.3\% & 95.7\% \\ [.2cm]

2 & $\pi_1=0.5$   &   0.502 (0.050) & 0.050  & 89.7\% & 95.7\%  \\
  &  $a_{12}=0.3$ &   0.300 (0.021) & 0.020  & 90.0\% & 94.3\% \\
  & $a_{21}=0.4$  &   0.401 (0.024) & 0.025  & 89.7\% & 95.3\% \\ [.2cm]
  3           & $\beta_0=2$   &   2.010 (0.173) & 0.177  & 91.0\% & 96.7\%   \\
& $\beta_1=5$   &   5.047 (0.357) & 0.364  & 90.7\% & 94.3\%   \\

\bottomrule
\end{tabular}
\label{tb:zparam:rep}
\end{center}
\end{table}

\begin{table}
\begin{center}
\caption{Data analysis results for \textit{iid} $z_{ik}$s  for $\matr{V}=\sigma^2\matr{I}$ and $\matr{V}_{\vect{z}_k}=\mbox{diag}(\sigma^2_{z_{1k}},\ldots,\sigma^2_{z_{nk}})$,
with corresponding fitted curves in Figures \ref{building_noHMM_equals2} and \ref{building_noHMM_diffs2}, respectively. Note that for $\matr{V}=\sigma^2\matr{I}$ the estimate $\hat{\sigma}^2$ does not depend on $j$ and, therefore, its value appears in the middle row.}
%\vspace{-.3cm}
\begin{tabular}{@{}l c c c c@{}}
  \toprule

 &  curve (chiller condition, $j$) & $\hat{\sigma}_j^2$ & $\hat{p}_j$ (SE) & $\lambda_j$  \\
  \midrule
 $\matr{V}=\sigma^2\matr{I}$ 
 &   black  (off, $j=1$)   & \multirow{2}{*}{103.5} & 0.665  (0.025) & 0.020 \\
  &  red  (on, $j=2$)   &  & 0.335  (0.025) & 0.078 \\
    \midrule
   $\matr{V}_{\vect{z}_k}=$
   & black (off, $j=1$)  & 12.7 & 0.658 (0.025) & 0.073 \\
 $\mbox{diag}(\sigma^2_{z_{1k}},\ldots,\sigma^2_{z_{nk}})$   & red  (on, $j=2$) & 355.4 & 0.342 (0.025) & 0.006 \\
\bottomrule
\end{tabular}
\label{tb:Building:iid}
\end{center}
\end{table}

\begin{table}
\begin{center}
\caption{Data analysis results for Markov $z_{ik}$s, for   $\matr{V}=\sigma^2\matr{I}$ and  $\matr{V}_{\vect{z}_k}=\mbox{diag}(\sigma^2_{z_{1k}},\ldots,\sigma^2_{z_{nk}})$, with corresponding fitted curves in Figures \ref{building_HMM_equals2} and \ref{building_HMM_diffs2}. Note that for $\matr{V}=\sigma^2\matr{I}$ the estimate $\hat{\sigma}^2$ does not depend on $j$ and, therefore, its value appears in the middle row. In addition, for the case when $\matr{V}_{\vect{z}_k}=\mbox{diag}(\sigma^2_{z_{1k}},\ldots,\sigma^2_{z_{nk}})$ we are not able to obtain SEs for $\hat{\pi}_1$,  $\hat{\pi}_2$, $\hat{a}_{12} $ and $\hat{a}_{21}$ as $\hat{a}_{21} < 10^{-16}$. }
\vspace{-.3cm}
\begin{tabular}{@{}l c c c c c c@{}}
  \toprule

  & &   &   & \multirow{2}{*}{\minitab[c]{$\hat{a}_{12}$ \\ (SE)}} & \multirow{2}{*}{\minitab[c]{$\hat{a}_{21}$ \\ (SE)}}  &   \\
&curve (chiller condition, $j$)& $\hat{\sigma}_j^2$ & $\hat{\pi}_j$ (SE) & & & $\lambda_j$   \\
  \midrule
$\matr{V}=\sigma^2\matr{I}$  &black (off, $j=1$)  & \multirow{2}{*}{103.4} & 0.705  (0.069) & \multirow{2}{*}{\minitab[c]{0.024 \\ (0.011)}} & \multirow{2}{*}{\minitab[c]{0.00991\\ (0.00986)}} & 0.019 \\
 &   red  (on, $j=2$)   &                          & 0.295  (0.069) &                        &                                       & 0.083 \\
    \midrule
$\matr{V}_{\vect{z}_k}=$
&    black (off, $j=1$)  & 12.2  & 0.682    & 0.015 &  $< 10^{-16}$ & 0.049 \\
$\mbox{diag}(\sigma^2_{z_{1k}},\ldots,\sigma^2_{z_{nk}})$
&    red   (on, $j=2$)  & 400.1 & 0.318   &         -          &               -                & 0.006 \\
\bottomrule
\end{tabular}
\label{tb:Building:Markov}
\end{center}
\end{table}

\begin{table}
\begin{center}
\caption{Data analysis results for $z_{ik}$s with distribution depending on a covariate (temperature) and $\matr{V}_{\vect{z}_k}=\mbox{diag}(\sigma^2_{z_{1k}},\ldots,\sigma^2_{z_{nk}})$ with corresponding fitted curves in Figure \ref{building_temperature}.}
%\vspace{-.3cm}
\begin{tabular}{@{}r c c c@{}}
  \toprule
 curve (chiller condition, $j$) & $\hat{\sigma}^2_j$ & $\hat{\beta}$ (SE) & $\lambda_j$  \\
  \midrule
    black (off, $j=1$)  &17.9  & $\hat{\beta}_0= -13.013$ (1.411)   & 0.115  \\
    red  (on, $j=2$) & 274.0 & $\hat{\beta}_1= 0.607$ (0.068) & 0.030 \\
\bottomrule
\end{tabular}
\label{tb:Building:cov}
\end{center}
\end{table}

\begin{table}[ht]
\begin{center}
\caption{Data analysis results for \textit{iid} $z_{ik}$s and $\matr{V}$ depending on the hidden states generated by a non-homogeneous random intercept model with $\lambda_1=\lambda_2=0.1$ chosen via cross-validation and corresponding curves in Figure \ref{building_RI_noHMM_complex}. Note that  $\hat{\sigma}^2$, $\hat{\tau}_1^2$ and $\hat{\tau}_2^2$ do not depend on $j$.}
%\vspace{-.3cm}
\begin{tabular}{@{}r c c c c@{}}
  \toprule
 curve (chiller condition, $j$) & $\hat{\sigma}^2$ & $\hat{\tau}_1^2$ & $\hat{\tau}_2^2$ & $\hat{p}_j$ (SE)   \\
  \midrule
    black  (off, $j=1$)   & \multirow{2}{*}{14.9} & \multirow{2}{*}{11.0} & \multirow{2}{*}{505.0} & 0.662  (0.025)  \\
    red  (on, $j=2$)   &  &  &  & 0.338  (0.025)  \\
\bottomrule
\end{tabular}
\label{tb:Building:iid:RI:complex}
\end{center}
\end{table}

\begin{figure}
        \centering
        \begin{subfigure}[b]{0.5\textwidth}
                \centering
                \includegraphics[width=\textwidth]{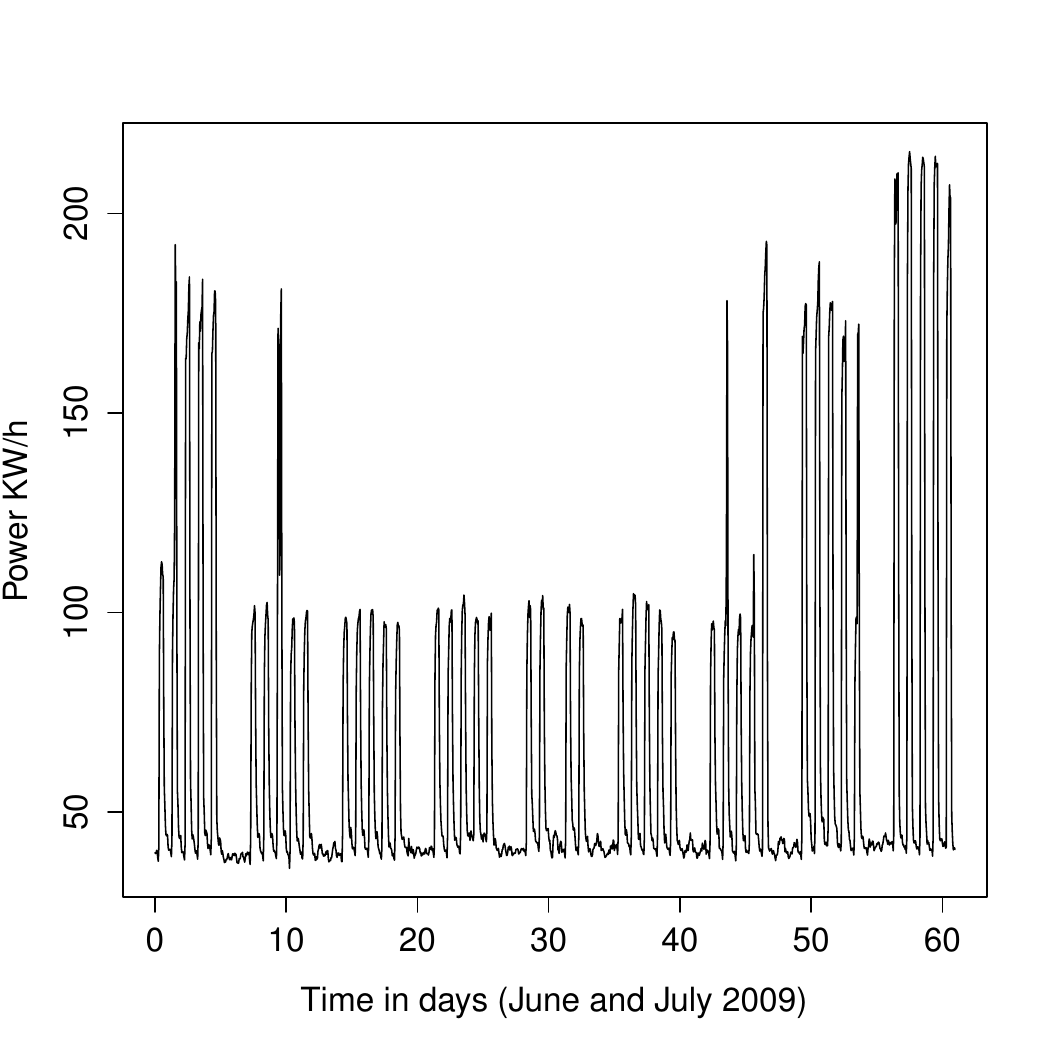}
                \caption{}
             \label{building_one}
        \end{subfigure}%
        \begin{subfigure}[b]{0.5\textwidth}
                \centering
                \includegraphics[width=\textwidth]{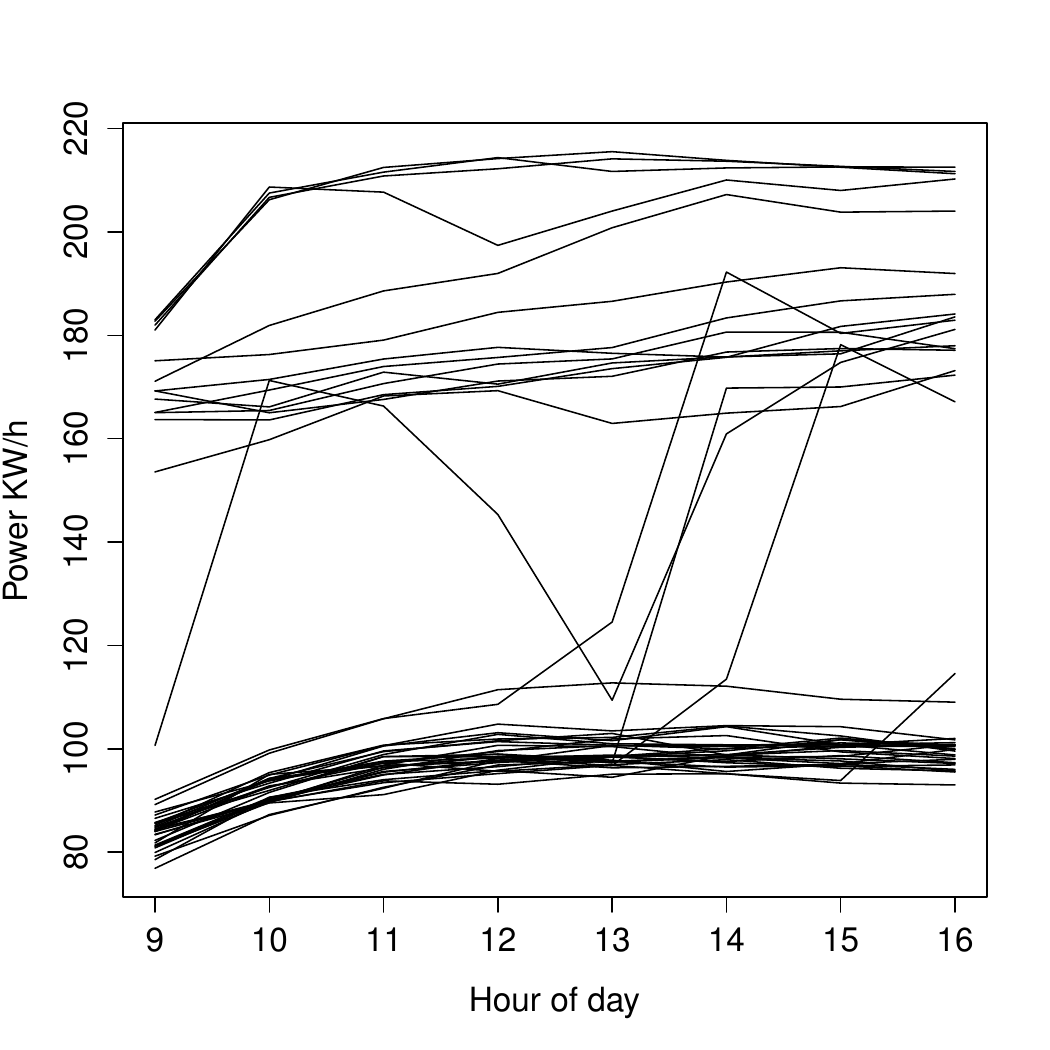}
             \caption{}
			\label{building_replicates}        
	\end{subfigure} 
\caption{ (a) Power usage in June and July 2009 in a building monitored by Pulse Energy. (b) Daytime power usage from 9am to 4pm on business days (each curve corresponds to a different day) in June and July 2009 for the same building.}
\label{building_figures}
\end{figure}

\begin{figure}
        \centering
        \begin{subfigure}[b]{0.5\textwidth}
                \centering
                \includegraphics[width=\textwidth]{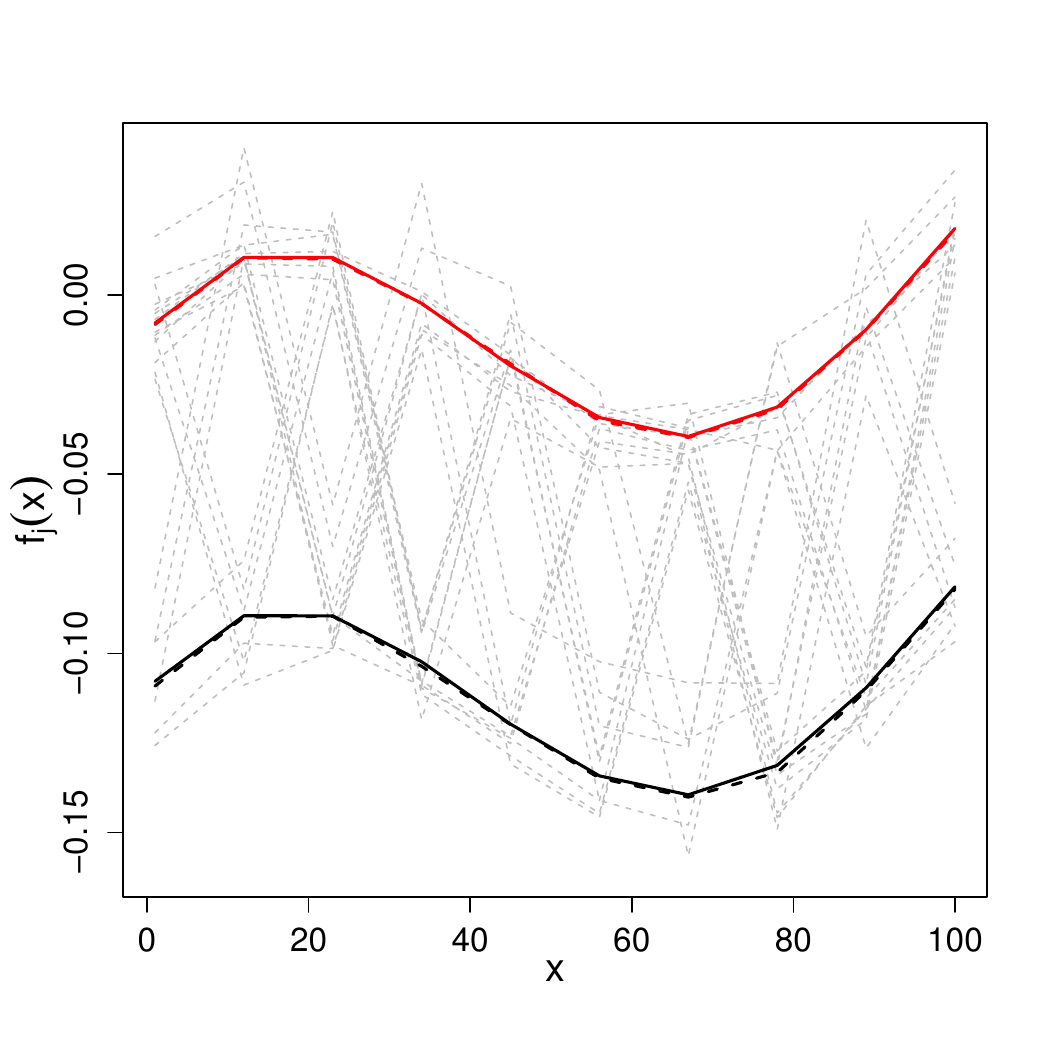}
                \caption{Design 1: \textit{iid} $z_{ik}$s}
             \label{S1:fitted:example}
        \end{subfigure}%
        \begin{subfigure}[b]{0.5\textwidth}
                \centering
                \includegraphics[width=\textwidth]{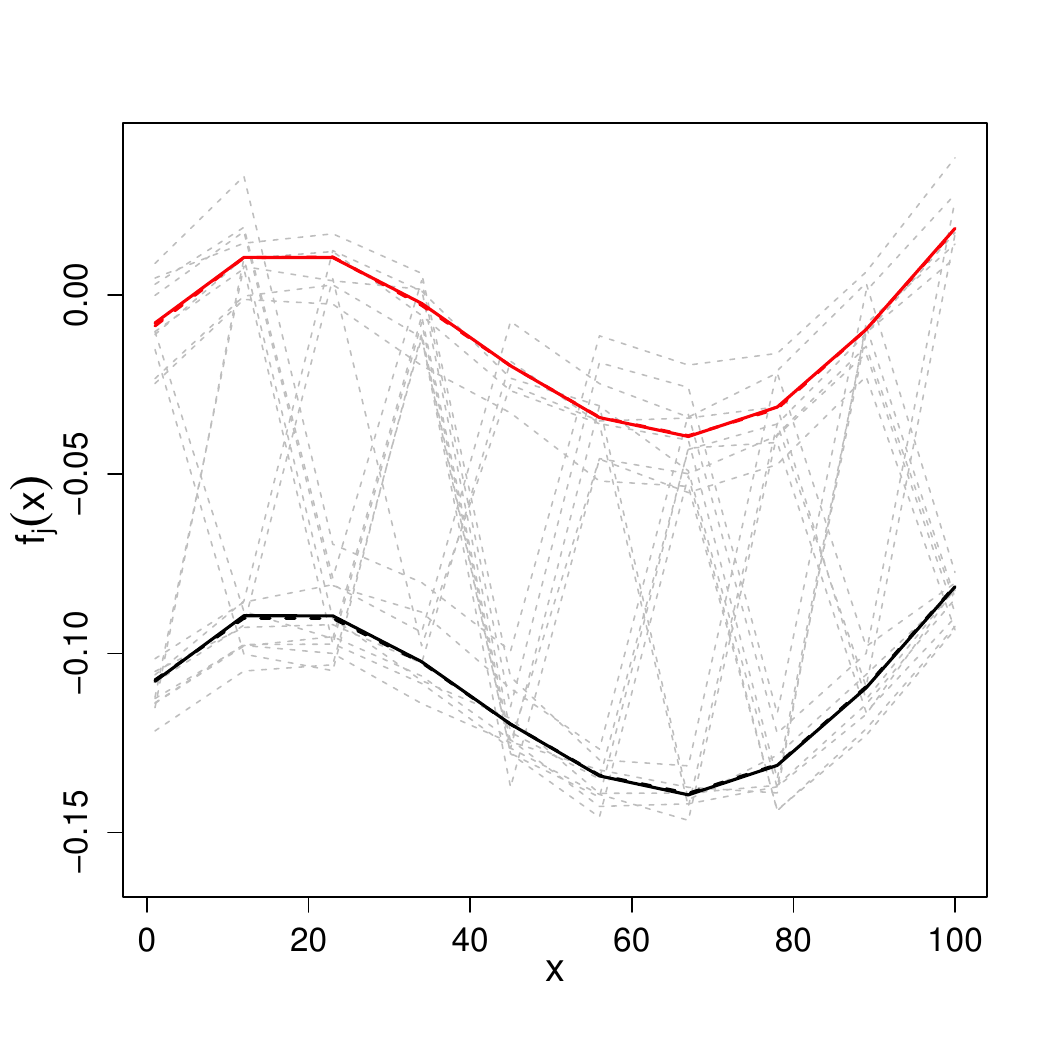}
                \caption{Design 2: Markov $z_{ik}$s}
                \label{S2:fitted:example}
        \end{subfigure} 
        \begin{subfigure}[b]{0.5\textwidth}
                \centering
                \includegraphics[width=\textwidth]{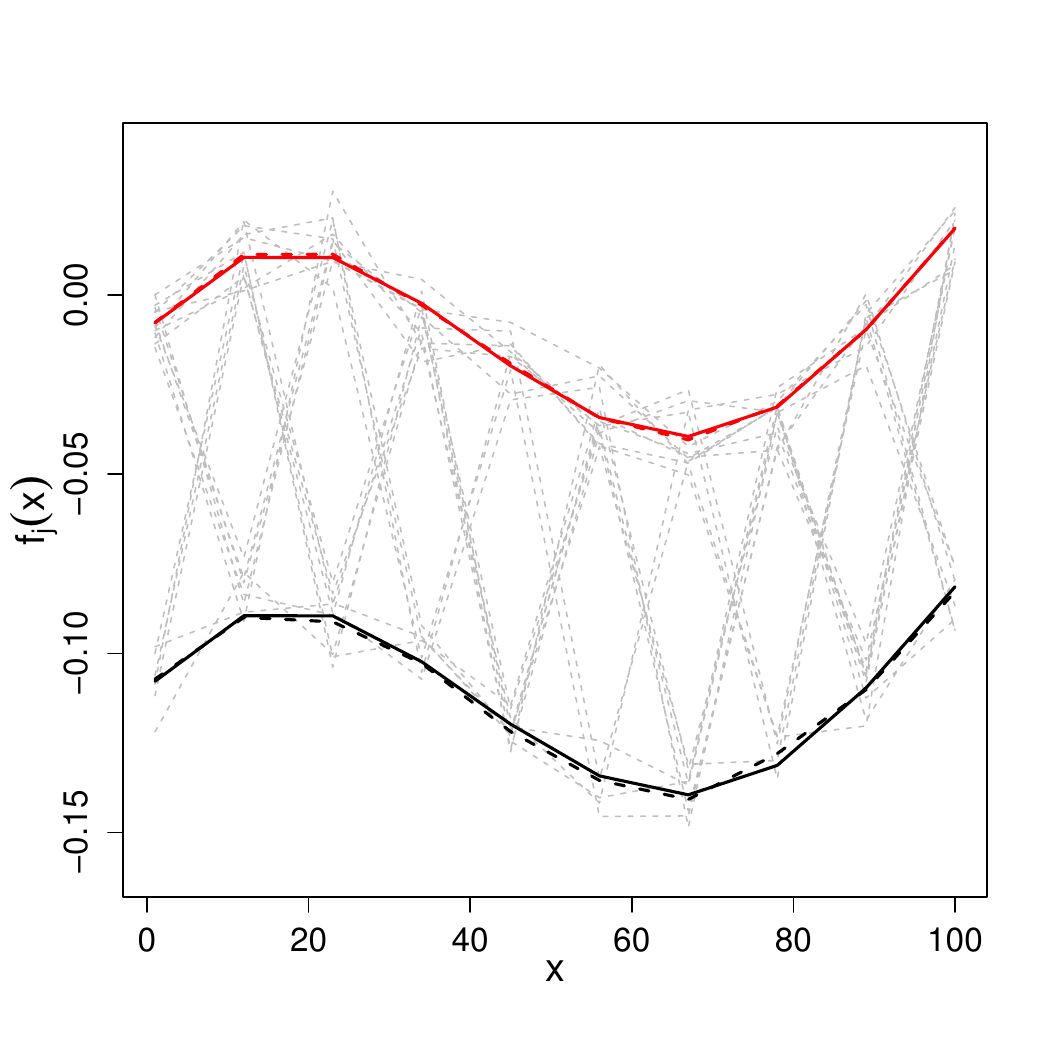}
                \caption{Design 3: covariate dependent $z_{ik}$s}
                \label{S3:fitted:example}
        \end{subfigure}
\caption{Example of simulated data along with $\hat{f}_1(\vect{x})$ and $\hat{f}_2(\vect{x})$ for each simulation design. The gray dashed curves correspond to 20 out of the 100 generated replicates. The black and red solid curves correspond to the true functions $f_1$ and $f_2$, respectively, evaluated only at $\vect{x}$. The black and red dashed curves correspond to $\hat{f}_1(\vect{x})$ and $\hat{f}_2(\vect{x})$, respectively.}
\label{Sim:fitted:example}
\end{figure}

\begin{figure}
        \centering
        \begin{subfigure}[b]{0.5\textwidth}
                \centering
                \includegraphics[width=\textwidth]{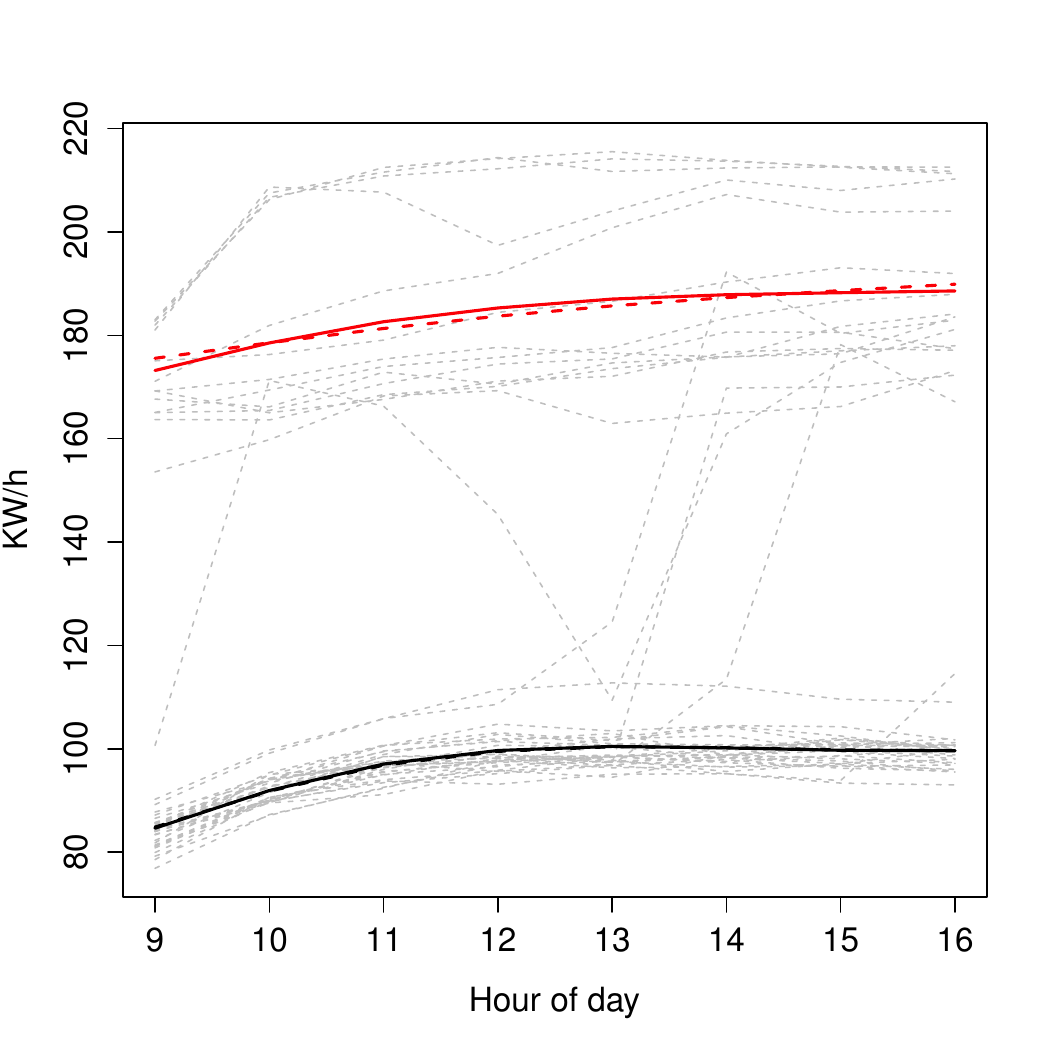}
                \caption{\textit{iid} $z_{ik}$s, $\matr{V}=\sigma^2\matr{I}$ }
                \label{building_noHMM_equals2}
        \end{subfigure}%
        \begin{subfigure}[b]{0.5\textwidth}
                \centering
                \includegraphics[width=\textwidth]{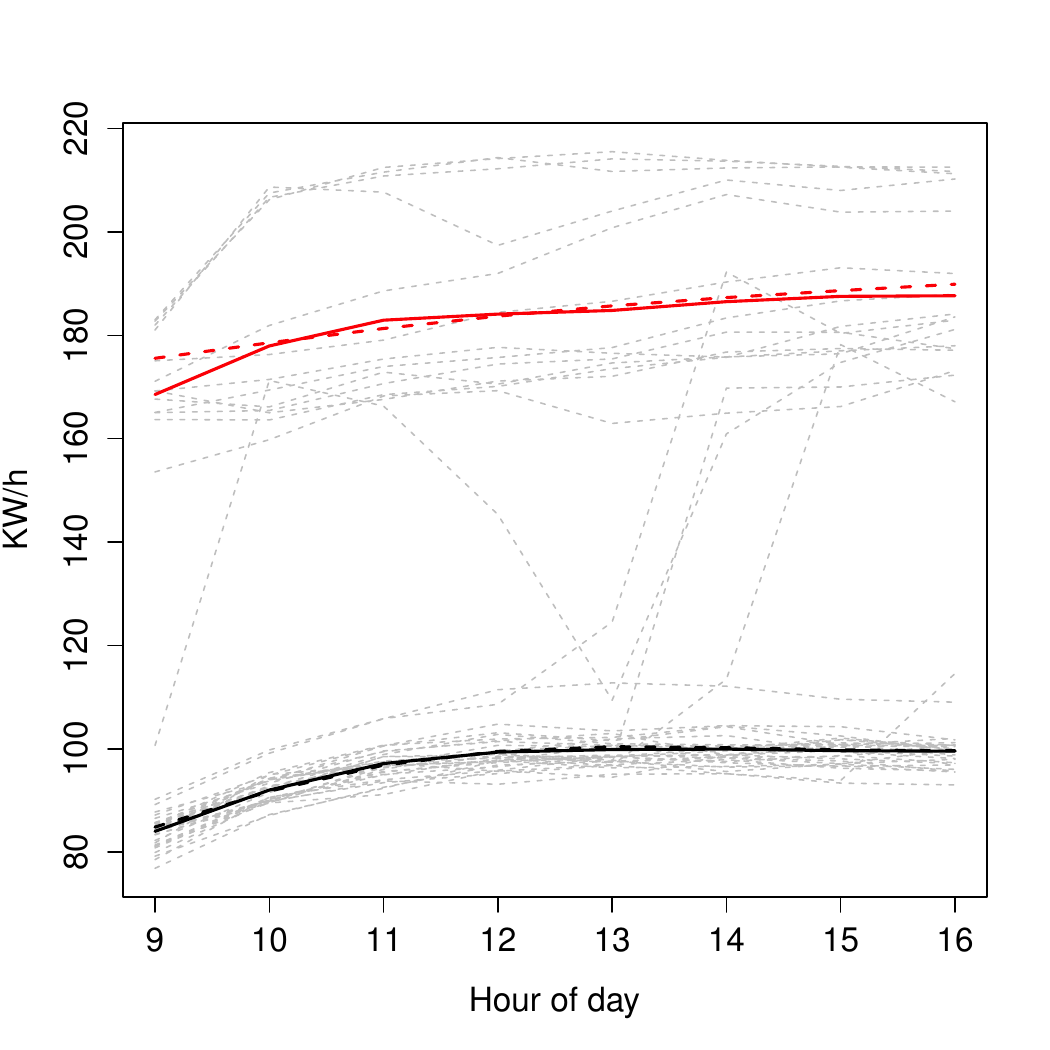}
                \caption{ \textit{iid} $z_{ik}$s$, \matr{V}_{\vect{z}_k}=\mbox{diag}(\sigma^2_{z_{1k}},\ldots,\sigma^2_{z_{nk}})$}
                \label{building_noHMM_diffs2}
        \end{subfigure}
        \begin{subfigure}[b]{0.5\textwidth}
                \centering
                \includegraphics[width=\textwidth]{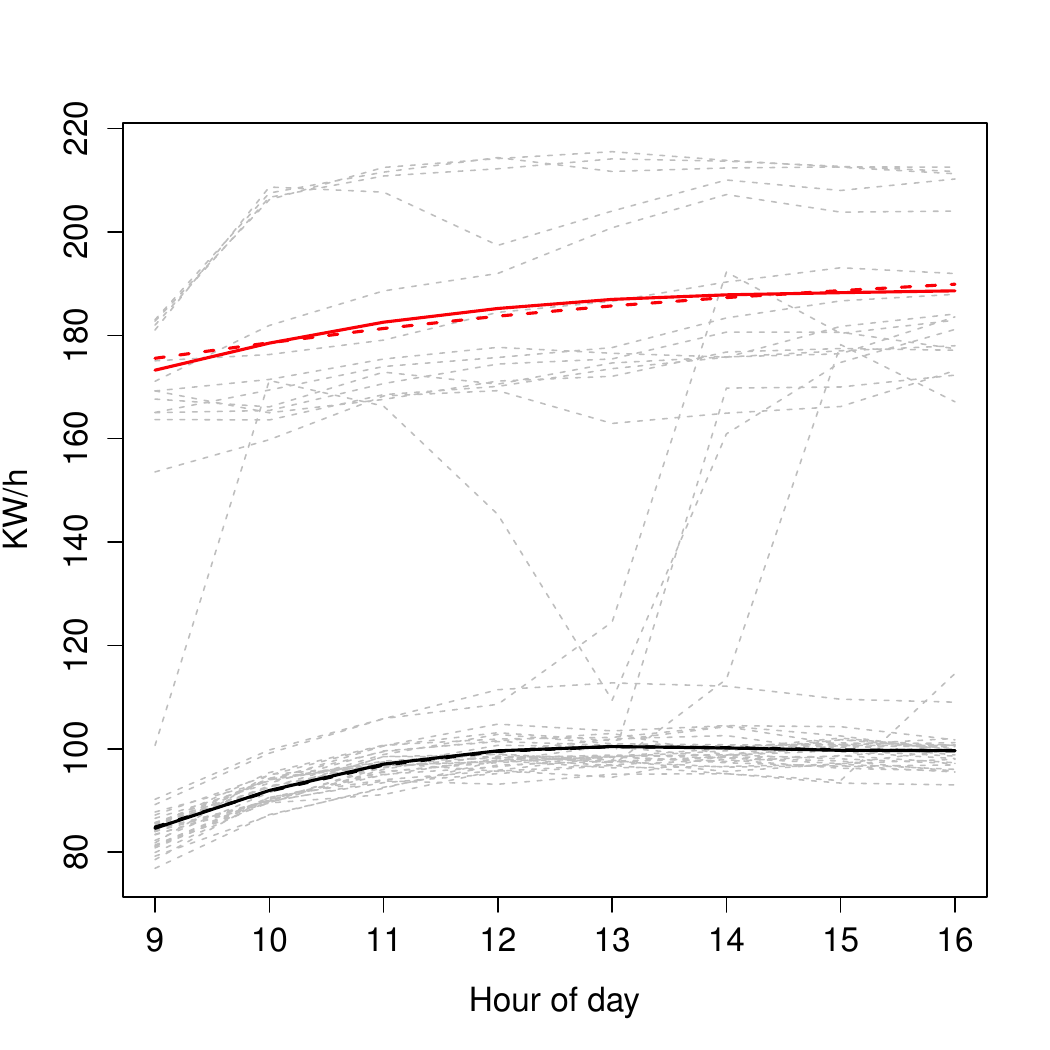}
                \caption{Markov $z_{ik}$s, $\matr{V}=\sigma^2\matr{I}$ }
              \label{building_HMM_equals2}
        \end{subfigure}%
        \begin{subfigure}[b]{0.5\textwidth}
                \centering
                \includegraphics[width=\textwidth]{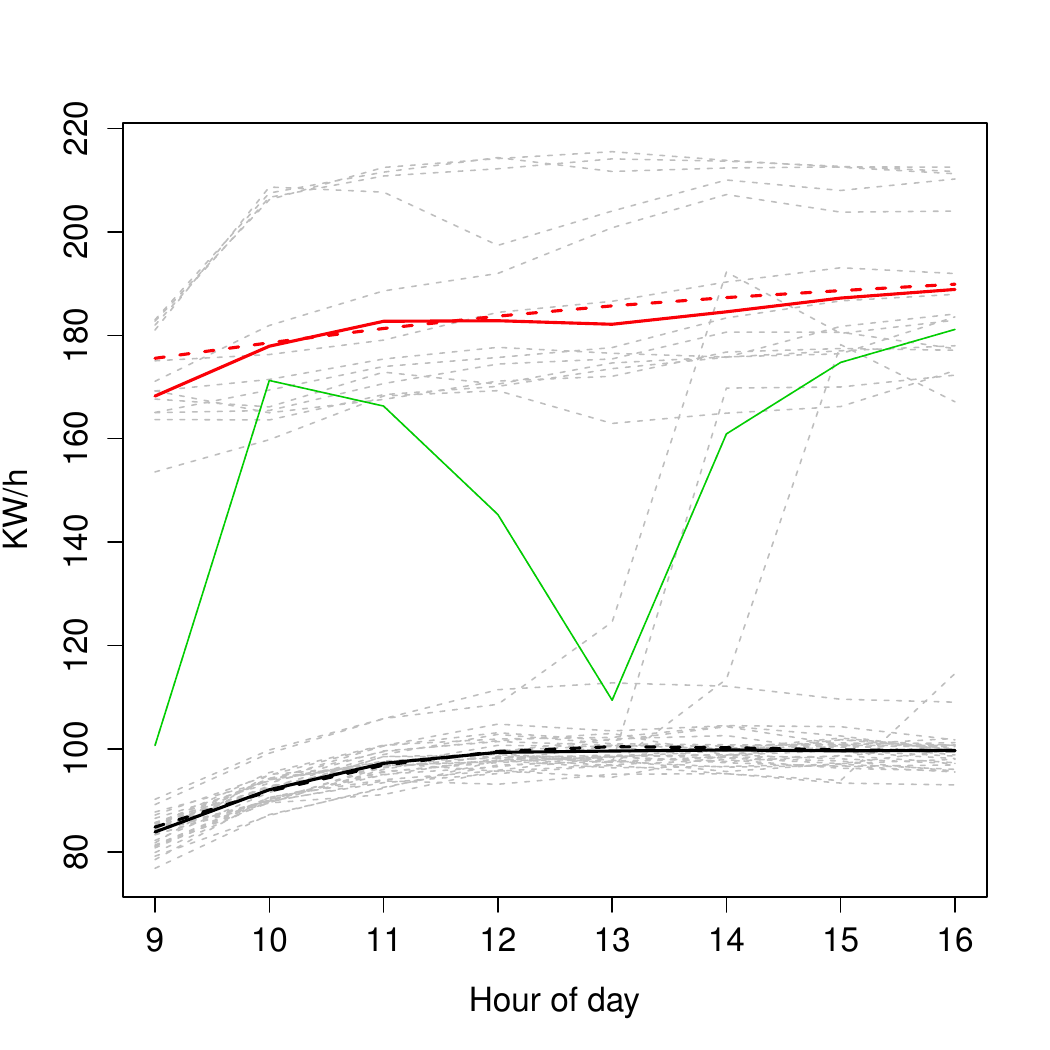}
                \caption{\footnotesize{Markov $z_{ik}$s, $\matr{V}_{\vect{z}_k}=\mbox{diag}(\sigma^2_{z_{1k}},\ldots,\sigma^2_{z_{nk}})$}}
             \label{building_HMM_diffs2}
        \end{subfigure}
\caption{Building daytime power usage. Fitted function estimates (solid curves) assuming \textit{iid} $z_{ik}$s (\textit{top row}) and Markov $z_{ik}$s (\textit{bottom row}). In (a) and (c) we consider $\matr{V}=\sigma^2\matr{I}$. In (b) and (d) $\matr{V}_{\vect{z}_k}=\mbox{diag}(\sigma^2_{z_{1k}},\ldots,\sigma^2_{z_{nk}})$.  The gray dashed curves correspond to the replicates. The red and black dashed curves are the initial function estimates. The colors red and black correspond to the condition chiller on and off, respectively. The green curve in (d) corresponds to the replicate where there is a transition from chiller on to off.}
\label{fig:building_not_correlated}
\end{figure}

\begin{figure}
        \centering
        \begin{subfigure}[b]{0.5\textwidth}
                \centering
                \includegraphics[width=\textwidth]{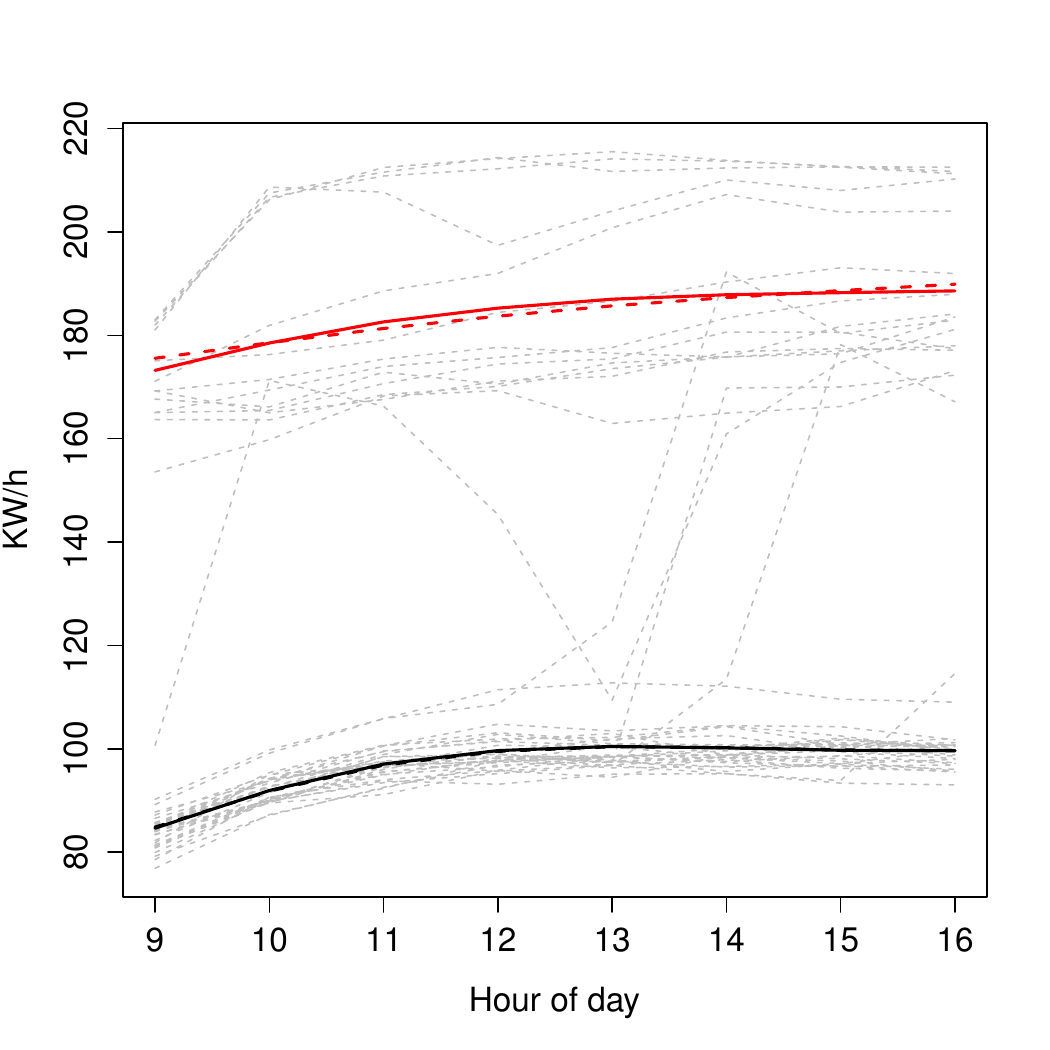}
                \caption{  $z_{ik}$s depending on temperature}
              \label{building_temperature}
        \end{subfigure}%
        \begin{subfigure}[b]{0.5\textwidth}
                \centering
                \includegraphics[width=\textwidth]{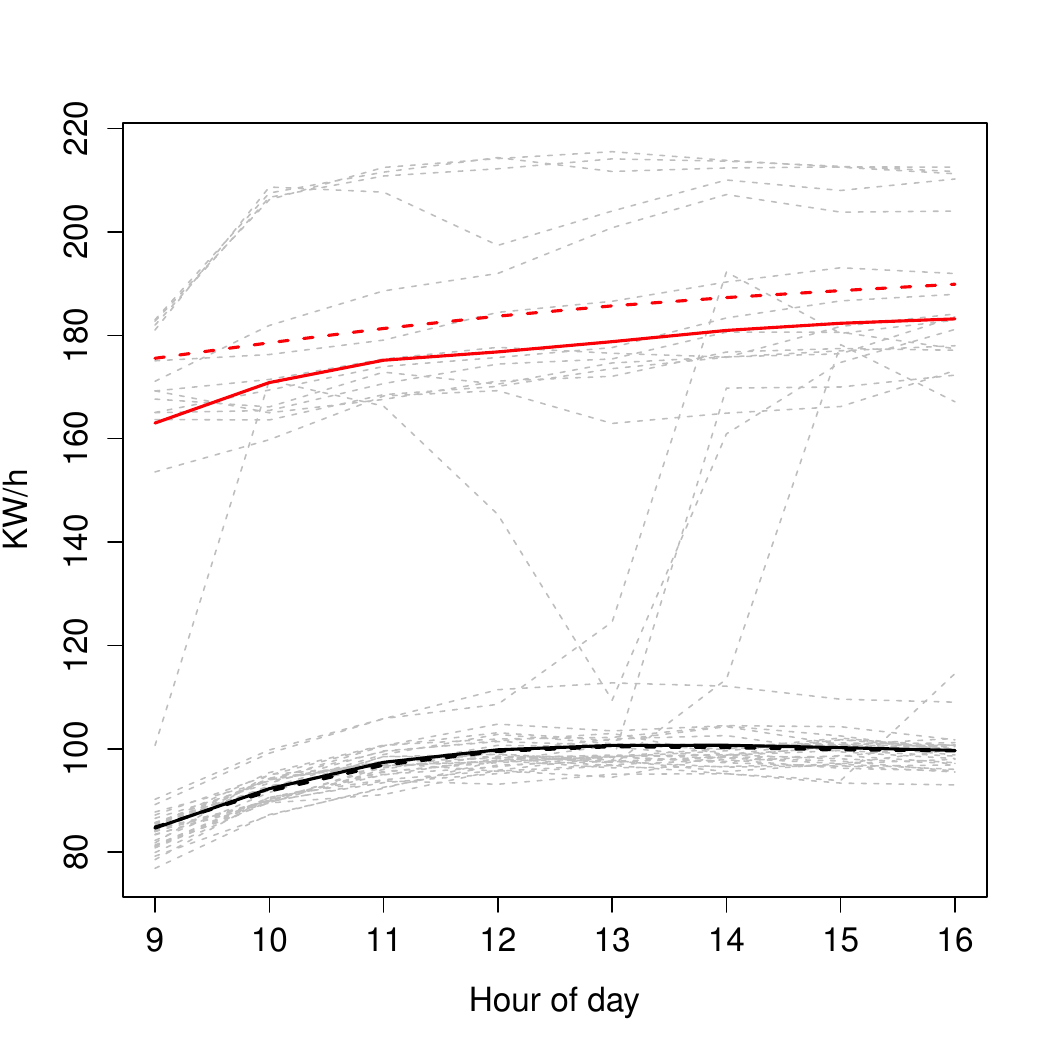}
                \caption{\textit{iid} $z_{ik}$s and $\matr{V}$ as in (\ref{covEq:RI:diff})}
             \label{building_RI_noHMM_complex}
        \end{subfigure}
\caption{Building daytime power usage. (a) Fitted function estimates assuming the $z_{ik}$s are independent with distribution depending on temperature. $\matr{V}_{\vect{z}_k}=\mbox{diag}(\sigma^2_{z_{1k}},\ldots,\sigma^2_{z_{nk}})$. (b) Fitted function estimates assuming \textit{iid} $z_{ik}$s and $\matr{V}$ generated by a non-homogeneous random intercept model as in (\ref{covEq:RI:diff}). Components of the plots are as in Figure \ref{building_noHMM_equals2}.}
\label{building_temp_and_RIcomplex}
\end{figure}

\end{document}